\documentclass{IEEEtaes}

\usepackage{color,array,amsthm}
\usepackage{graphicx}

\jvol{XX}
\jnum{XX}
\jmonth{XXXXX}
\paper{1234567}
\pubyear{2025}
\doiinfo{TAES.2025.Doi Number}

\newtheorem{theorem}{Theorem}
\newtheorem{lemma}{Lemma}
\newtheorem{remark}{Remark}

\newtheorem{definition}{Definition}
\usepackage{caption,booktabs,amsthm,amsmath,threeparttable,subfigure,float,amssymb}

\begin{document}

\title{Spacecraft Safe Robust Control Using Implicit Neural Representation for Geometrically Complex Targets in Proximity Operations} 

\author{Hang Zhou}
\affil{School of Aeronautics and Astronautics, Zhejiang University, China} 

\author{Tao Meng}
\affil{School of Aeronautics and Astronautics, Zhejiang University, China\\ Hainan Research Institute of Zhejiang University, China} 

\author{Kun Wang}
\affil{Hangzhou Innovation Institute, Beihang University, Hangzhou, China} 

\author{Chengrui Shi}
\affil{School of Aeronautics and Astronautics, Zhejiang University, China} 

\author{Renhao Mao}
\affil{School of Aeronautics and Astronautics, Zhejiang University, China} 

\author{Weijia Wang}
\affil{School of Aeronautics and Astronautics, Zhejiang University, China} 

\author{Jiakun Lei}
\affil{International Innovation Institute, Beihang University, Hangzhou, China} 

\authoraddress{
Hang Zhou, Chengrui Shi, Renhao Mao, Weijia Wang are with the School of Aeronautics and Astronautics, Zhejiang University, Hangzhou 310027, China, E-mail: (zhou\_hang@zju.edu.cn; chengruishi@zju.edu.cn; maorh@zju.edu.cn; weijiawang@zju.edu.cn).  
Tao Meng is with the School of Aeronautics and Astronautics, Zhejiang University, Hangzhou 310027, China, and also with the Hainan Institute of Zhejiang University, Sanya 572025, China, E-mail: (mengtao@zju.edu.cn).  
Kun Wang is with the Hangzhou Innovation Institute, Beihang University, Hangzhou 310051, China (Email: wang\_kun@zju.edu.cn). 
Jiakun Lei is with the International Innovation Institute, Beihang University, Hangzhou 311115, China (Email: leijiakun@buaa.edu.cn). 
The research was Supported by the China Postdoctoral Science Foundation under Grant Number 2025M774251 and the Zhejiang Provincial Natural Science Foundation of China under Grant No. LR24F030001. 
(Corresponding Author: Tao Meng; Kun Wang.)}

\markboth{Zhou ET AL.}{Preprint Submitted To IEEE Transactions On Aerospace and Electronic Systems}
\maketitle

\begin{abstract}
This paper investigates the problem of safe control for spacecraft proximity operations, with a particular emphasis on the collision avoidance of a chaser spacecraft relative to a target spacecraft with complex geometry in the presence of disturbances. To ensure safety in such scenarios, a safe robust control framework is proposed that leverages implicit neural representations. To handle arbitrary target geometries without explicit modeling, we learn a neural signed distance function (SDF) from point cloud data via a enhanced implicit geometric regularization method, which incorporates an over-apporximation strategy to create a conservative, safety-prioritized boundary. The zero-level set of the learned neural SDF implicitly defines the target's surface, while the values and gradients provide critical information for safety controller design. This neural SDF representation underpins a two-layer hierarchcial safe robust control framework: a safe velocity generation layer and a safe robust controller layer. In the first layer, a second-order cone program is formulated to generate safety-guaranteed reference velocity by explicitly incorporating the under-approximation error bound. Furthermore, a circulation inequality is introduced to mitigate the local minimum issues commonly encountered in control barrier function (CBF) methods. In the second layer, a disturbance observer is incorporated, and a smooth safety filter is designed by explicitly accounting for the observer's estimation error, thereby enhancing robustness against external disturbances. Extensive numerical simulations and Monte Carlo analysis validate the proposed framework, demonstrating significantly improved safety margins and avoidance of local minima compared to conventional CBF approaches. 
\end{abstract}


\section{INTRODUCTION}
With the continued advancement and deployment of space stations and other high-value spacecraft, the demand for efficient maintenance and servicing operations has become increasingly critical \cite{li_orbit_2019}. Recent technological progress, particularly in the development of small autonomous spacecraft, has opened new possibilities for cost-effective and reliable solutions \cite{ReddBringing}. These spacecraft can enhance mission autonomy and reduce operational costs, enabling a shift toward more sustainable and scalable servicing architectures. Among these emerging capabilities, proximity operations—which involve precise and safe relative motion in close-range environments—play a foundational role in applications such as satellite inspection \cite{day2020two,bohan2023parametric,faghihi2023multiple}, debris removal \cite{forshaw2017final, forshaw2016removedebris}, and rendezvous and docking \cite{fehse2003automated}. 

A central challenge in proximity operations lies in ensuring the safety of the chaser spacecraft, particularly by preventing collisions with the target during close-range maneuvers. The difficulty of this task is significantly amplified when the target spacecraft exhibits complex or irregular surface geometrie. Such spacecraft may lack standard structural symmetry, present protrusions like antennas, solar panels and robotic arms. These factors increase the risk of unexpected surface contact and complicate the design of safe, reactive, and computationally efficient  control strategies. Therefore, accurately capturing the target's geometry and incorporating that information into the control loop is essential for enabling safe  autonomous proximity operations.

To address the safety problem in spacecraft proximity operations, a variety of theoretical approaches have been developed. The most widely adopted methods can be broadly categorized into path planning, optimization-based control, barrier function-based methods, and reference governor (RG) methods. Path planning methods typically utilize search algorithms to construct feasible trajectories in the configuration space \cite{frey2017constrained,karaman2011sampling,deka2023astrodynamics}, however, their reliance on global search often leads to high computational costs and limited real-time applicability. Optimization-based methods reformulate the safety problem as a constrained optimal control problem, embedding safety requirements as constraints and solving for an optimal policy through numerical optimization \cite{zhang2023stochastic,li2020spacecraft}. Among these, Model Predictive Control (MPC) \cite{weiss2015model,bashnick2023fast,wang2023model} has emerged as a prominent optimization-based technique for spacecraft maneuvering under constraints. 
Barrier function-based methods encompass a range of techniques such as Artificial Potential Function (APF) \cite{zhou2023collision, wang2022artificial, hu2021velocity,hwang2022collision}, Prescribed Performance Control (PPC) \cite{li2023prescribed,wu2022adaptive}, and Barrier Lyapunov Function (BLF) \cite{tee2009barrier,huang2019adaptive}, each providing analytical tools for enforcing state constraints and ensuring safety during system evolution. An alternative approach for constraint handling is the RG \cite{bemporad2002reference, gilbert2002nonlinear}. In particular, the Explicit Reference Governor (ERG) \cite{hu2023spacecraft} relies on a guidance strategy typically constructed using the APF method, thereby inheriting the inherent limitations associated with APF-based approaches.

To address the limitations of traditional methods, particularly the trade-off between safety guarantees and computational efficiency, recent research has focused on formal methods that can provide rigorous safety certificates. Control barrier function (CBF) have emerged as a particularly promising framework, offering several advantages over traditional approaches: they provide formal safety guarantees through forward invariance properties, enable real-time implementation through efficient quadratic programming formulation \cite{ames2019control}, and allow for systematic handling of multiple safety constraints \cite{molnar_composing_2023, breeden_compositions_2023}. 
In spacecraft control applications, recent studies have explored various implementations of CBF. These works utilize simplified geometric representations,such as ellipsoids and capsules\cite{shi_safe_2025, wang_control_2024, wang2025safe}, and construct safety constraints using softmin and softmax approaches \cite{molnar2023composing}. 
Nevertheless, significant challenges remain in applying CBF methods to spacecraft proximity operations. A key difficulty lies in constructing continuously differentiable barrier functions that can accurately account for target spacecraft with extremely complex geometries. Existing methods often rely on simplifying assumptions and are often overly conservative in describing complex shapes.

The signed distance function (SDF) describes how far a point is from the boundary of a given set, while also indicating whether the point lies inside or outside the set \cite{boczko2005signeddistancefunctionnew}. As an implicit representation of surfaces, SDF was widely used in computer vision and graphics for tasks such as surface reconstruction and rendering \cite{Park_2019_CVPR,lin2020sdf}. Compared with other geometric representations, SDF offers direct access to both distance and gradient information relative to object surfaces, which is particularly beneficial for constructing CBF. 

Representing shapes as zero-level sets of neural networks learned SDF. we call such representations implicit neural representation (INR). Recently, impressive results have been achieved in object shape modeling using deep neural networks. DeepSDF \cite{Park_2019_CVPR} and Occupancy Nets \cite{mescheder2019occupancy} implicitly represent three-dimensional (3D) shapes by supervised learning via fully connected neural networks. Gropp et al. \cite{gropp2020implicit} further extended the DeepSDF model by introducing the Eikonal constraint to computing high fidelity INR directly from raw data. To promote conservative geometry approximation suitable for safety-critical applications, we design a loss function that explicitly favors over-approximation of the object's surface. Specifically, the loss assigns asymmetric penalties to the predicted signed distance values by combining their raw and absolute values with a tunable parameter. This encourages the network to produce non-negative outputs around the true surface, effectively pushing the learned zero-level set outward from the actual geometry. Consequently, the learned implicit surface forms a safe outer envelope that over-approximates the true object boundary.

This paper proposes a series of innovative technical solutions to the problem of safe close-range operation control of complex structure target spacecraft. The main contributions of this paper are summarized as follows:

\begin{enumerate}
    \item A deep neural SDF learning method with an over-approximation strategy is proposed to implicitly represent the complex geometry of the target spacecraft. By intentionally overestimating the SDF, the resulting implicit surface forms a conservative yet accurate outer envelope of the target. This representation exhibits enhanced accuracy and adaptability compared to traditional explicit geometric modeling techniques.
    \item A safety filter based on second-order cone programming (SOCP) is developed to generate velocity references with formal safety guarantees. The approach explicitly considers the worst-case approximation error of the learned SDF and incorporates relaxed circulation inequalities (CI) to enable the chaser spacecraft to escape from local minima.
    \item A smooth safety controller incorporating a disturbance observer (DO) is further designed to ensure safety under external disturbances. By actively compensating for disturbance effects, the proposed control strategy significantly reduces the conservatism of the admissible safety set. 
\end{enumerate}

The rest of this paper is organized as follows: the notations and necessary definitions are given in Section \ref{PreliminariesAndNotations} and the spacecraft model and relative position dynamics are given in Section \ref{modelAndDynamic}. The constraints are illustrated and then the control objective are given in Section \ref{ProblemFormulation}. The INR of the target spacecraft is presented in Section \ref{ImplicitNeuralRepresentation}, where the SDF learning method and loss function are described. The safe robust control framework design is presented in Section \ref{RobustSafeBacksteppingControllerDesign}, which includes the safety velocity generation layer and the safe robust controller layer. The simulation results are presented in Section \ref{simulation}, demonstrating the effectiveness of the proposed control framework. Finally, conclusions are given in Section \ref{conclusion}.

\section{PRELIMINARIES AND NOTATIONS}
\label{PreliminariesAndNotations}
\subsection{Notations}
This paper uses the following standard notations: $\mathcal{C}^{m}(\mathcal{X},\mathbb{R}^m)$ denotes the space of $m$-times continuously differentiable functions mapping $\mathcal{X} \in \mathbb{R}^n \to \mathbb{R}^n$; $\boldsymbol{a}^{\times}$ denotes the skew symmetric matrix form of $\boldsymbol{a} = [a_1,a_2,a_3]^{\top}$; $|\cdot{}|$ denotes the absolute value of a scalar; $\left\|\cdot\right\|$ denotes the Euclidean norm of a vector; $\boldsymbol{I}_{n\times n}$ denotes the identity matrix of size $n \times n$; $\boldsymbol{0}$ denotes the zero vector with suitable dimension; a continuous function $\alpha: [0,a) \to [0,\infty)$ is said to belong to class $\mathcal{K}_{\infty}$ if $\lim_{r\to a} \alpha(r) = \infty$ , $\alpha(0) = 0$, and it is strictly increasing; a function $\alpha: [-b,a)\to[-\infty, \infty),a > 0, b > 0 $, belongs to extended class $\mathcal{K}_{\infty}$ if it is strictly increasing and $\alpha(0) = 0$.

\subsection{Definitions}
\begin{definition}
    \label{Def_SDF}
    Signed Distance Function (SDF, \cite{boczko2005signeddistancefunctionnew}): Let $\Omega$ be a subset of $\mathbb{R}^n$, and $\partial \Omega$ be its boundary. The distance between a point $\boldsymbol{p} \in \mathbb{R}^n$ and the subset $\partial \Omega$ is defined as $d(\boldsymbol{p},\partial \Omega)$, the SDF from a point $\boldsymbol{p}$ of $\mathbb{R}^n$ to $\Omega$ is defined by:
    \begin{equation}
        \label{def_sdf}
        f(\boldsymbol{p}) = \begin{cases}
            d(\boldsymbol{p},\partial \Omega) & \text{if } \boldsymbol{p} \notin \Omega \\
            0 & \text{if } \boldsymbol{x} \in \partial \Omega \\
            -d(\boldsymbol{p},\partial \Omega) & \text{if } \boldsymbol{p} \in \Omega,
        \end{cases}
    \end{equation}
    and its gradient satisfies the eikonal equation $\|\nabla f \| = 1$
\end{definition}

\begin{definition}
    \label{dcbf}
    Control Barrier Function (CBF, \cite{ames_control_2017}):
    Consider the nonlinear affine control system $\dot{\boldsymbol{x}} = f(\boldsymbol{x})+g(\boldsymbol{u})$, where $f$ and $g$ are locally Lipschitz, $\boldsymbol{x}\in \mathcal{X} \subset\mathbb{R}^n$ and $\boldsymbol{u}\in U\subset\mathbb{R}^m$. A continuously differentiable function $h \in \mathcal{C}^{1}(\mathcal{X})$ is called a CBF if there exists an extended class $\mathcal{K}_{\infty}$ function $\alpha$ such that for all $\boldsymbol{x}\in \mathcal{X}$:
    \begin{equation}
        \label{cbf_condition_inequality}
        \sup_{\boldsymbol{u}\in U}\left[L_fh(\boldsymbol{x})+ L_gh(\boldsymbol{x})\boldsymbol{u}\right] \geq-\alpha(h(\boldsymbol{x})).
    \end{equation}
\end{definition}


\begin{definition}
    \label{ISS_CLF}
    Input-to-state stable Control Lyapunov Function (ISS-CLF, \cite{kolathaya_input_2018}): A continuously differential and positive definite function V is an ISS-CLF for system $\dot{\boldsymbol{x}} = f(\boldsymbol{x})+g(\boldsymbol{u}+\boldsymbol{d})$ with disturbance $\boldsymbol{d}$, if there exist functions $\alpha_2 \in \mathcal{K}_{\infty}$ and $\alpha_3 \in \mathcal{K}_{\infty}$ such that for all $\boldsymbol{x}$:
    \begin{equation}
        \begin{aligned}
            &\inf_{\boldsymbol{u} \in U}\left\{L_fV(\boldsymbol{x}) + L_gV(\boldsymbol{x})(\boldsymbol{u} + \boldsymbol{d})\right\} \\
            &\leq -\alpha_2(\|V(\boldsymbol{x})\|) + \alpha_3(\|\boldsymbol{d}\|).
        \end{aligned}
    \end{equation}
\end{definition}

\section{DYNAMICS SYSTEM MODELING}
\label{modelAndDynamic}
\subsection{Coordinate Frames}
As illustrated in Fig.(\ref{fig_frames}), the following coordinate systems are defined: The Earth-centered inertial (ECI) frame is denoted by $\{\mathcal{I}\}: O_E \boldsymbol{X}_{E}\boldsymbol{Y}_{E}\boldsymbol{Z}_{E}$. The body frame of the target spacecraft, denoted as $\{\mathcal{B}_{T}\}$, is also introduced, with its origin located at the spacecraft's center of mass. The coordinate axes of each frame are aligned with the spacecraft's principal axes of inertia and follow the right-hand rule. Additionally, the target's orbital frame, denoted as $\{\mathcal{O}_{T}\}$, is defined with its origin at the target's center of mass. The $\boldsymbol{X}_{\mathcal{O}_{T}}$ axis is aligned with the target's velocity vector, while the $\boldsymbol{Z}_{\mathcal{O}_{T}}$ axis points towards the center of Earth. The $\boldsymbol{Y}_{\mathcal{O}_{T}}$ axis, along with the $\boldsymbol{X}_{\mathcal{O}_{T}}$ and $\boldsymbol{Z}_{\mathcal{O}_{T}}$, forms a right-hand coordinate system. This frame is also commonly referred to as the vehicle velocity and local horizontal (VVLH) coordinate system. For computational convenience, assuming that the body frame $\{\mathcal{B}_{T}\}$ coincides with the orbital frame $\{\mathcal{O}_{T}\}$ at all times. Unless otherwise specified, vectors are expressed in the orbital frame $\{\mathcal{O}_{T}\}$.

\begin{figure}[hbt!]
    \centering
    \includegraphics[scale=0.15]{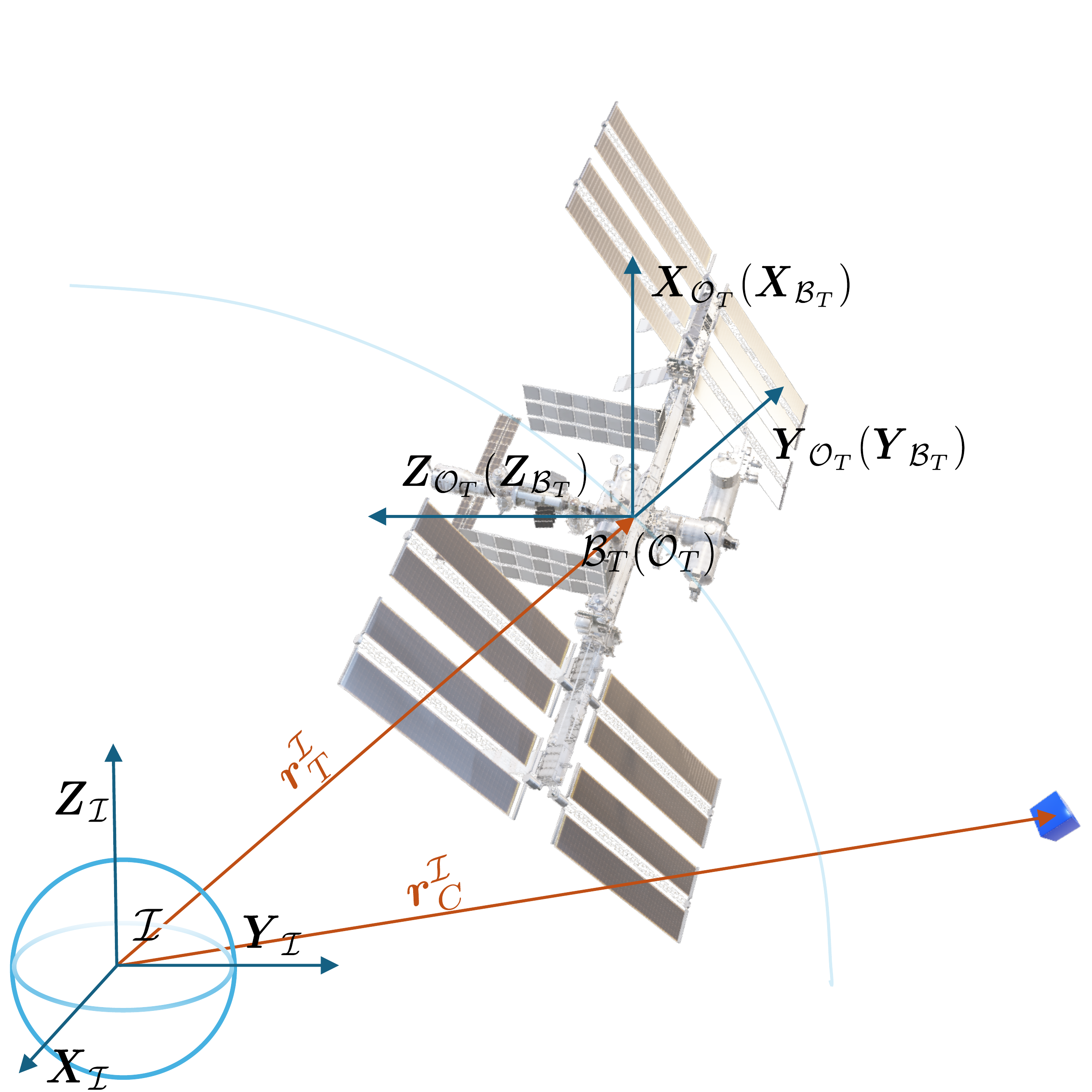}
    \caption{Illustration of coordinate systems}
    \label{fig_frames}
\end{figure}

\subsection{Relative Position Dynamics} 
Assuming the target is in a specific orbit and the chaser is positioned nearby, let $\boldsymbol{r}, \boldsymbol{v}$ denote the position and velocity of the chaser relative to the target. Furthermore, let $\boldsymbol{r}_{S\mathcal{I}}$ and $\boldsymbol{r}_{T\mathcal{I}}$ denote the positions of the chaser and the target, respectively, in frame $\{\mathcal{I}\}$. The relative orbital dynamics of the spacecraft can then be described as \cite{yamanaka2002new}:

\begin{equation}
    \label{relative_orbital_dynamics_equation}
m \ddot{\boldsymbol{r}} + m \boldsymbol{C}_1 \dot{\boldsymbol{r}} + m \boldsymbol{C}_2 \boldsymbol{r} + m \boldsymbol{g} = \boldsymbol{F} + \boldsymbol{d}
\end{equation}
where $\boldsymbol{C}_1 = 2 \boldsymbol{\omega}_o^{\times}$, $\boldsymbol{C}_2 = \boldsymbol{\dot{\omega}}_o^{\times} + \boldsymbol{\omega}_o^{\times} \boldsymbol{\omega}_o^{\times}$, and $\boldsymbol{g} = \mu \left( \frac{\boldsymbol{r}_{S\mathcal{I}}}{\|\boldsymbol{r}_{S\mathcal{I}}\|_{2}^{3}} - \frac{\boldsymbol{r}_{T\mathcal{I}}}{\|\boldsymbol{r}_{T\mathcal{I}}\|_{2}^{3}} \right)$. The term $\boldsymbol{d}$ represents external disturbances force, while $\boldsymbol{F}$ denote the position control force for the service. The parameter $m$ represents the mass of the chaser. The $\boldsymbol{\omega}_o = \begin{bmatrix}0 & -f_{\theta} & 0\end{bmatrix}^\top$, where $f_{\theta}$ denotes the true anomaly of the target. For arbitrary orbits $\dot{f}_{\theta} = \sqrt{\frac{\mu}{{a^{3}(1 - e^{2})^{3}}}} \left(1 + e \cos(f_{\theta})\right)^{2}$ , $\ddot{f}_{\theta} = - \frac{2 \mu e \sin(f_{\theta}) \left(1 + e \cos(f_{\theta})\right)^{3}}{{a^{3}(1 - e^{2})^{3}}}$. $a$ denotes the semimajor axis of the target, $e$ denotes the eccentricity of the target, and $\mu$ is the gravitational constant. Then Eq.(\ref{relative_orbital_dynamics_equation}) can be rewritten as:

\begin{subequations}
\begin{align}
    \dot{\boldsymbol{r}} &= \boldsymbol{v}  \label{kinematic_system}\\
    \dot{\boldsymbol{v}} &= -\boldsymbol{C}_1\boldsymbol{v} - \boldsymbol{C}_2\boldsymbol{r} + \boldsymbol{g} + \frac{1}{m}\left(\boldsymbol{F} + \boldsymbol{d} \label{dynamics_system} \right).
\end{align}
\label{relative_orbital_dynamics}
\end{subequations} 

\section{PROBLEM FORMULATION} 
\label{ProblemFormulation}
The primary objective of this work is to develop a safe robust control framework guaranteeing operational safety during proximity maneuvering missions around irregularly complex shaped target spacecraft. 

\subsection{Constraints Illustration and Representation}

\subsubsection{Safety Constraints}
During the mission, the chaser must satisfy a safety constraint that ensures it does not collide with the target. For target with complex structures, it is challenging to approximate their geometry using regular shapes such as cuboids, capsules, or ellipsoids. To address this challenge, this paper employs a SDF to implicitly represent the complex surface of the target, enabling more detailed and accurate characterization of its intricate geometry. As Definition \ref{Def_SDF}, the condition $f_{sdf} = 0$ implicitly defines the boundary of a set. For a 3D object, this boundary corresponds exactly to the object's surface. 

Therefore, in order to ensure that the chaser does not collide with the target, the collision constraint can be naturally written as: 
\begin{equation}
    \label{relative_position_constraint}
    h(\boldsymbol{r}) = f_{sdf}(\boldsymbol{r}) \geq 0.
\end{equation}

Moreover, the SDF satisfies the following property:
\begin{equation}
    \label{relative_position_constraint_gradient}
    \nabla h(\boldsymbol{r}) = \| \nabla f_{sdf}(\boldsymbol{r}) \| = 1. 
\end{equation}

\subsubsection{Relative Velocity Constraint}
Given the low kinetic energy requirements for the chaser during the mission, velocity constraints must be imposed to ensure that the spacecraft operates within acceptable limits. These velocity constraints can be expressed as follows:
\begin{equation}
    \label{relative_velocity_constraint}
\boldsymbol{v}=\begin{bmatrix}v_x &v_y&v_z\end{bmatrix}^\top  \in \mathbb{R}^3, \left|v_{i}\right| \leq v_{\text{max}}, i = x, y, z.
\end{equation}

\subsection{Control Problem Formulation}
The control framework is illustrated in Fig.(\ref{control_framework}), and the objective of this work is to design a control law $\boldsymbol{F}$ for the relative orbital dynamics described by Eq.(\ref{relative_orbital_dynamics}) such that the chaser converges from its initial position $\boldsymbol{r}_s$ to the desired position $\boldsymbol{r}_d$, while satisfying the constraints in Eq.(\ref{relative_position_constraint}) and Eq.(\ref{relative_velocity_constraint}), and accounting for a bounded external disturbance $\boldsymbol{d}$. 

\begin{remark}
    \label{remark_velocity_constraint}
    The velocity constraints Eq.(\ref{relative_velocity_constraint}) are treated as soft constraints, meaning that their violation is permitted but should be taken into consideration during the generation of a safe reference velocity. This is because safety is explicitly defined in terms of collision avoidance in position space, whereas strict adherence to velocity constraints is not essential for ensuring safety.
\end{remark}

\begin{figure*}[hbt!]
    \centering
    \includegraphics[scale=0.55]{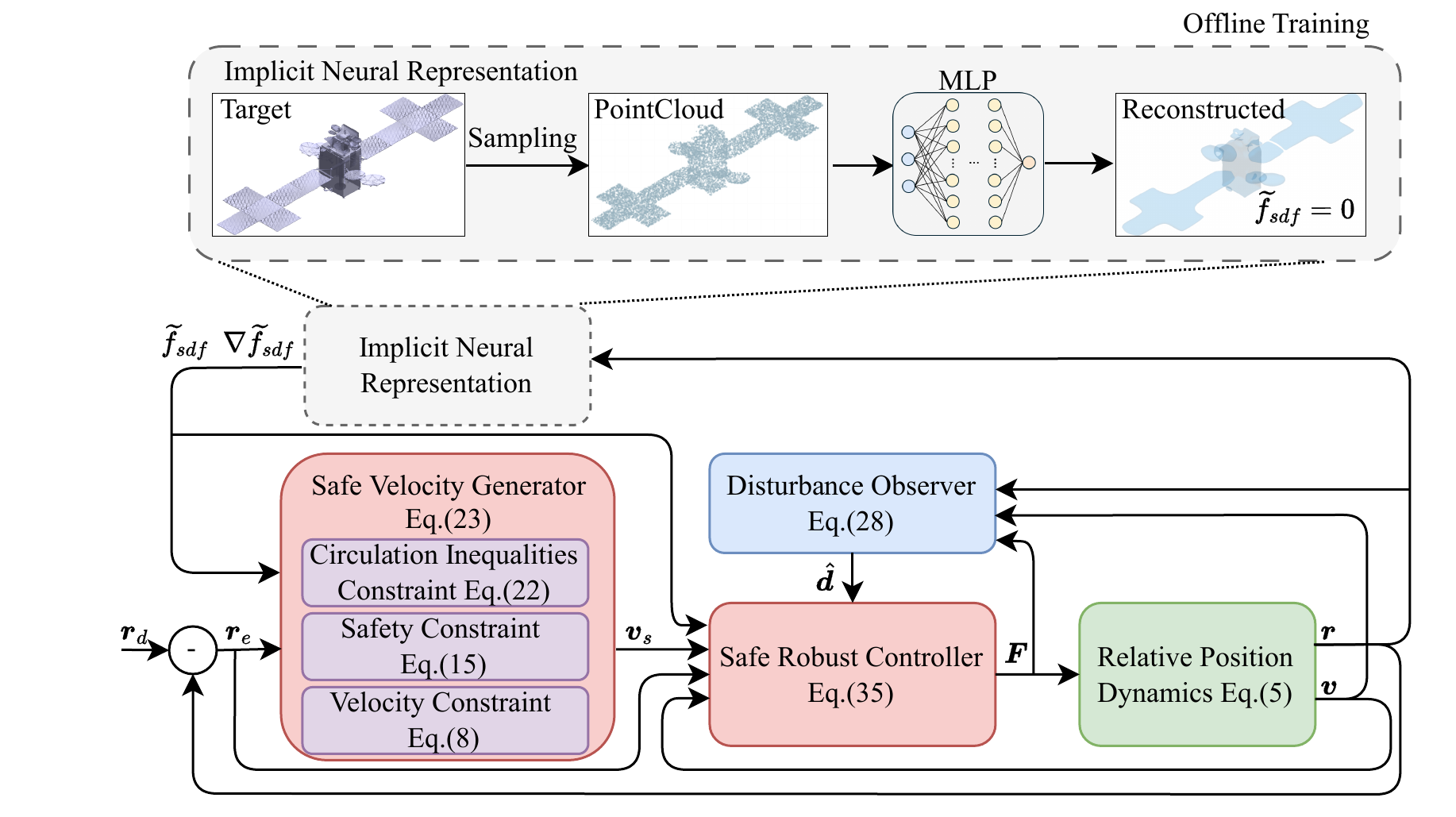}
    \caption{Illustration of control framework}
    \label{control_framework}
\end{figure*}

\section{IMPLICIT NEURAL REPRESENTATION}
\label{ImplicitNeuralRepresentation}
Inspired by existing research on multilayer perceptron (MLP) approximation of SDF \cite{gropp2020implicit}, \cite{parkDeepSDFLearningContinuous2019}, we employ a fully connected neural network $\widetilde{f}_{sdf}$ to approximate the SDF of target structure. To enhance geometric fidelity and reduce surface ambiguity, an improved over-approximation strategy is employed. This strategy operates by deliberately overestimating the SDF, thereby pushing the zero-level set outward from the true surface. As a result, the learned implicit surface forms a conservative envelope that safely encloses the target geometry.

\subsection{Data Generation} 
In this work, the shape of the target spacecraft is assumed to be known. Accordingly, a dataset $\mathcal{P}$ is prepared, consisting of pairs of 3D point samples and their corresponding SDF values:
\begin{equation}
    \mathcal{P} = \left\{\left(\boldsymbol{p}_i, d_i\right)\right\}_{i=1}^{I} \subset \mathbb{R}^3 \times \mathbb{R}, i = 1,2,\ldots,I
\end{equation} 
where $\boldsymbol{p}_i$ represents the $i$-th point in the point cloud, and $d_i$ is the signed distance value corresponding to $\boldsymbol{p}_i$. 

\subsection{Loss Function}
The loss function is defined as $\mathcal{L} = \mathcal{L}_P + \eta  \mathcal{L}_E$, with a parameter $\eta  > 0$ and $\mathcal{L}_P$ and $\mathcal{L}_E$ defined as follows: 

\begin{equation}
    \begin{aligned}
        \label{loss_function_LP}
        &\mathcal{L}_P = \frac{1}{I} \sum_{i=1}^{I} \left( \frac{\kappa- 1}{2} (\widetilde f_{sdf}(\boldsymbol{p}_i)-d_i) + \frac{\kappa - 1}{2} |\widetilde f_{sdf}(\boldsymbol{p}_i)-d_i| \right) \\
        &\mathcal{L}_E = \frac{1}{I} \sum_{i=1}^{I} \left( \| \nabla \widetilde f_{sdf}(\boldsymbol{p}_i)\| - 1 \right)^2
    \end{aligned}
\end{equation}
where, $\mathcal{L}_P$ encourages that the network output SDF error is close to zero at the training data points, $\kappa$ is the weighting coefficients for positive and negative results respectively, used to control the overestimation and underestimation of the approximation results, $\mathcal{L}_E$ is called the $Eikonal$ term encourages that the gradient of the network output SDF $ \nabla \widetilde f_{sdf}$ is close to unit length at the training data points. By adjusting the value of $\eta$, the influence of these two loss terms can be balanced. 

\subsection{Computing Gradients}
Since the loss function incorporates gradient terms derived from the network's output, and the subsequent controller design must account for gradients of the learned SDF, the development of efficient gradient computation methods emerges as a crucial consideration. 

Each layer of the MLP $\widetilde f_{sdf}$ is formulated as $\boldsymbol{y}^{\ell+1} = \boldsymbol{\tau}\left(\boldsymbol{W}^{\ell}\boldsymbol{y}^{\ell} + \boldsymbol{b}^{\ell}\right) $. where $\boldsymbol{\tau} : \mathbb{R} \to \mathbb{R}$ denotes a nonlinear differentiable activation function, with $\boldsymbol{W}^{\ell}$ and $\boldsymbol{b}^{\ell}$ representing the layer's trainable weight matrix and bias vector, respectively. Consequently, through the chain rule of differentiation, the gradients satisfy 
\begin{equation}
    \nabla_{\boldsymbol{x}}\boldsymbol{y}^{\ell+1}=\operatorname{diag}\left( {\boldsymbol{\tau}}^{\prime}\left(\boldsymbol{W}^{\ell}\boldsymbol{y}^{\ell}+\boldsymbol{b}^{\ell}\right)\right)\boldsymbol{W}^{{\ell}}\nabla_{\boldsymbol{x}}\boldsymbol{y}^{\ell}
\end{equation}
where $\boldsymbol{\tau} ^{\prime}$ denotes the derivative of the activation function. In practical implementations, automatic differentiation facilitates efficient computation of the partial derivatives $\nabla \widetilde{f}_{sdf}(\boldsymbol{x})$.

\subsection{Worst Case Estimation}
To design a safety-guaranteed controller, we must analyze the training results to quantify the worst-case approximation of the learned SDF. For safety controller design, over-approximation is acceptable (albeit resulting in conservative outcomes), while under-approximation risks compromising safety. As shown in the Fig.(\ref{approximation}), two different approximations of $f_{sdf} = 0$ are shown. The over-approximation result on the left is more conservative than the actual value. At this time, when $\widetilde{f}_{sdf} > 0$, the original $f_{sdf}$ must also be greater than 0. The right is an example of under-approximation. When $\widetilde{f}_{sdf} > 0$, it is not necessarily guaranteed that $f_{sdf}>0$. 
Since the designed loss function does not strictly guarantee an over-approximation, it is necessary to analyze the under-approximation behavior of the trained model. This allows us to determine the worst-case approximation bound of the SDF, thereby ensuring safety guarantees in a rigorous manner.

\begin{figure}[hbt!]
    \centering
    \includegraphics[scale=0.15]{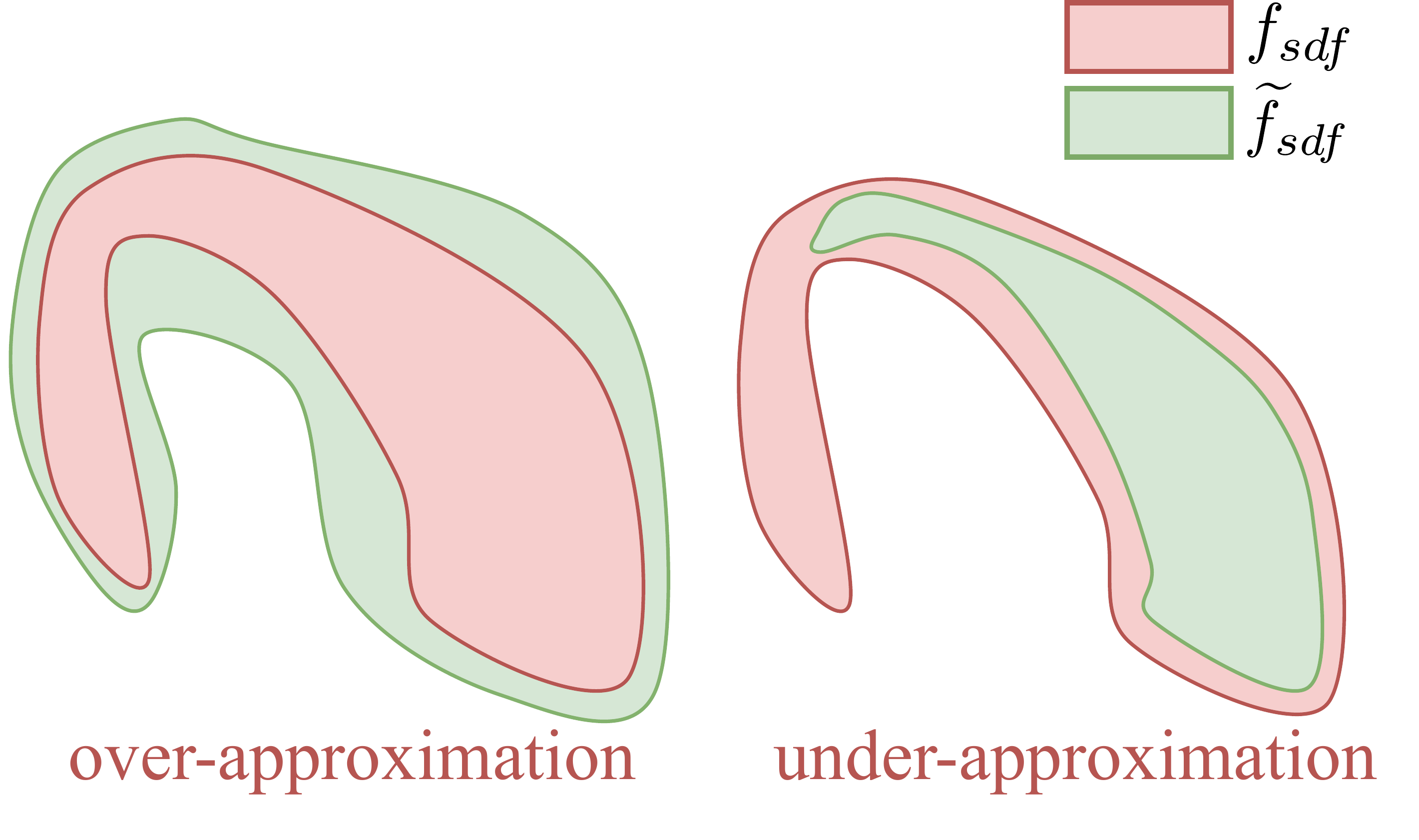}
    \caption{Illustration of approximation}
    \label{approximation}
\end{figure}

For the estimated SDF $\widetilde{f}_{sdf}(\boldsymbol{p}) \in \mathbb{R}$, the under-approximation estimation errors with respect to the ground-truth SDF $f_{sdf}(\boldsymbol{p})$ are defined as:
\begin{subequations}
\label{estimation_error}
\begin{align}
    e_1(\boldsymbol{p}) &:= \min(0,f_{sdf}(\boldsymbol{p}) - \widetilde{f}_{sdf}(\boldsymbol{p}))  \\
    \boldsymbol{e}_2(\boldsymbol{p}) &:= \nabla f_{sdf}(\boldsymbol{p}) - \nabla \widetilde{f}_{sdf}(\boldsymbol{p}). 
\end{align}
\end{subequations}
And the estimation error bound satisfies
\begin{equation}
\label{sdf_error_bound}
    \begin{aligned}
    |e_1(\boldsymbol{p})| &\leq {e}_h \\
    \|\boldsymbol{e}_2(\boldsymbol{p})\| &\leq {e}_{\nabla h}.
    \end{aligned}
\end{equation}

\section{SAFE ROBUST CONTROL FRAMEWORK DESIGN}
\label{RobustSafeBacksteppingControllerDesign}
When the INR of the target is established, a safe robust control framework is developed in this section through backstepping synthesis. The main idea behind backstepping is to treat the states of lower layers as virtual control inputs to the top layer, and then design a virtual controller for the top layer that would accomplish the given control objective.

The proposed control framework consists of two components: safe velocity generation and safe robust controller. For safe velocity generation, a safety filter including CI constraints and SDF collision constraints considering estimation errors is constructed in combination with the kinematic system to generate the safe reference velocity as the virtual control. For the safe robust controller, we integrate a DO with a safe backstepping control architecture. A smooth safety filter is employed to synthesize a robust control law that ensures system robust safety guarantees.

\subsection{Safe Velocity Generation}
As described in Section \ref{ImplicitNeuralRepresentation}, the zero-level set of the approximation SDF represents the surface of the spacecraft, while considering the worst-case approximation of the SDF. This makes it possible to formally formulate the safety constraint Eq.(\ref{relative_position_constraint}). 

Thus, the collosion avoidance constraint can be formulated as $h(\boldsymbol{r}) = f_{sdf} = \widetilde{f}_{sdf}(\boldsymbol{r}) + e_1(\boldsymbol{r}) \geq 0$, and the safe set can be written as 
\begin{equation}
    \mathcal{C} = \{\boldsymbol{r} \in \mathbb{R}^3, h(\boldsymbol{r}) \geq 0\}.
\end{equation}

\begin{theorem}
    Let 
\begin{equation}
    \begin{aligned}
    \label{vs_gengeration}
    &\mathcal{V}_s := \{ \boldsymbol{v}_s \in \mathbb{R}^3 :  \\
                    &\nabla {\widetilde{f}_{sdf}(\boldsymbol{r})}  \boldsymbol{v}_s -  \|\boldsymbol{v}_s\| e_{\nabla h} \geq - \alpha_0(\widetilde{f}_{sdf}(\boldsymbol{r}) - e_{h})  \}. 
    \end{aligned}
\end{equation}
$\boldsymbol{v}_s \in \mathcal{V}_s$ reflects a virtual controller for system Eq.(\ref{kinematic_system}) that guarantees the safe set $\mathcal{C}$ is forward invariant.
\end{theorem}

\begin{proof}
    As $h(\boldsymbol{r}) = \widetilde{f}_{sdf}(\boldsymbol{r}) + e_1(\boldsymbol{r})$ in Eq.(\ref{cbf_condition_inequality}):
    \begin{equation}
        \label{cbf_filter_constraint_original}
        \nabla \widetilde{f}_{sdf} \boldsymbol{v}_s + \boldsymbol{e}_2(\boldsymbol{r})^{\top} \boldsymbol{v}_s \geq -\alpha_0(\widetilde{f}_{sdf}(\boldsymbol{r}) + e_{1}(\boldsymbol{r})).
    \end{equation}

    For any fixed $\boldsymbol{r}$ and any errors $e_1(\boldsymbol{r})$ and $\boldsymbol{e}_2(\boldsymbol{r})$ satisfying Eq.(\ref{sdf_error_bound}), we need the minimum value of the left-hand side greater than the maximum value of the right-hand side to ensure that Eq.(\ref{cbf_filter_constraint_original}) still holds, namely:
    \begin{equation}
        \label{cbf_socp_minmax}
        \begin{aligned}
            &\min_{\|\boldsymbol{e}_2(\boldsymbol{r})\| \leq e_{\nabla h} } \{\nabla \widetilde{f}_{sdf} \boldsymbol{v}_s + \boldsymbol{e}_2(\boldsymbol{r})^{\top} \boldsymbol{v}_s \} \geq  \\
            &\max_{|e_1(\boldsymbol{r})| \leq e_h} \{-\alpha_0(\widetilde{f}_{sdf}(\boldsymbol{r}) + e_{1}(\boldsymbol{r}))\}.
        \end{aligned}
    \end{equation}

    Note that since $e_h \geq 0$ and $\alpha_0$ is an extended class $\mathcal{K}_{\infty}$ function, the maximum value in Eq.(\ref{cbf_socp_minmax}) is obtained by:
    \begin{equation}
        \label{cbf_socp_max}
        \max_{|e_1(\boldsymbol{r})| \leq e_h} \{-\alpha_0(\widetilde{f}_{sdf}(\boldsymbol{r}) + e_{1}(\boldsymbol{r}))\} = -\alpha_0(\widetilde{f}_{sdf}(\boldsymbol{r}) - e_h).
    \end{equation}

    The minimum value is attained by setting $\boldsymbol{e}_2(\boldsymbol{r}) = -\nabla \widetilde{f}_{sdf}(\boldsymbol{r}) e_{\nabla h}$, which leads to:
    \begin{equation}
        \label{cbf_socp_min}
        \min_{\|\boldsymbol{e}_2(\boldsymbol{r})\| \leq e_{\nabla h} } \{\nabla \widetilde{f}_{sdf} \boldsymbol{v}_s + \boldsymbol{e}_2(\boldsymbol{r})^{\top} \boldsymbol{v}_s \} = \nabla {\widetilde{f}_{sdf}(\boldsymbol{r})}  \boldsymbol{v}_s -  \|\boldsymbol{v}_s\| e_{\nabla h}.
    \end{equation}

    Substituting Eq.(\ref{cbf_socp_max}) and Eq.(\ref{cbf_socp_min}) into Eq.(\ref{cbf_socp_minmax}), we obtain the CBF constraint Eq.(\ref{vs_gengeration}). 
\end{proof}

Therefore, the safe velocity can be obtained by solving the following second-order cone program (SOCP) problem:

\begin{equation}
    \label{cbf_socp0}
    \begin{aligned}
        &\min_{\boldsymbol{v}_s} \quad \frac{1}{2} \|\boldsymbol{v}_c - \boldsymbol{v}_s\|^2 \\
        &\text{s.t.} 
        \begin{aligned}
        & \quad \nabla {\widetilde{f}_{sdf}(\boldsymbol{r})} \boldsymbol{v}_s + \alpha_0(\widetilde{f}_{sdf}(\boldsymbol{r}) - e_{h}) \geq \|\boldsymbol{v}_s\| e_{\nabla h} 
        \end{aligned}
    \end{aligned}
\end{equation}
where $\boldsymbol{v}_c = -v_{\max}\mathrm{Tanh}(\boldsymbol{r}_e/k_p)$, with $\boldsymbol{r}_e = \boldsymbol{r} - \boldsymbol{r}_d$, denotes the reference velocity which reflects a stabilizing virtual controller for system Eq.(\ref{kinematic_system}). $k_p$ is a positive constant that adjusts the convergence speed.

Unfortunately, the formulation in Eq.(\ref{cbf_socp0}) is prone to have stable equilibrium points besides the only desired one $\boldsymbol{r} = \boldsymbol{r}_d$. Those equilibrium point is called a local minimum.
To solve this problem, we refer to the method of \cite{10472718}, introduce the concept of CI, and introduce an inequality term in the SOCP problem to guide the chaser out of the local minimum point to achieve the final convergence goal $\boldsymbol{r}_d$. The form of CI is as follows:

\begin{equation}
    T(\boldsymbol{r})^{\top} \boldsymbol{v}_s \geq \Upsilon(h(\boldsymbol{r}))
\end{equation}
where $\Upsilon: \mathbb{R}^+ \to \mathbb{R}$ is a continuously decreasing function that satisfies $\Upsilon(0) > 0$ , $\Upsilon(h) \to -\infty$ when $h \to \infty$, and $T(\boldsymbol{r}) = \boldsymbol{\Omega} \nabla \widetilde{f}_{sdf}(\boldsymbol{r})$. $\boldsymbol{\Omega}$ is a matrix to be designed.

In order to meet the feasibility of the SOCP problem, we add parameters to relax the CI constraints:
\begin{equation}
    \quad T(\boldsymbol{r})^{\top} \boldsymbol{v}_s  - \Upsilon(\widetilde{f}_{sdf}(\boldsymbol{r}) - e_{h}) \geq \sigma.
\end{equation}

After adding the CI, we have the following SOCP problem we called CCBF-SOCP:
\begin{equation}
    \label{cbf_socp_CI}
    \begin{aligned}
        &\min_{\boldsymbol{v}_s,\sigma} \quad \frac{1}{2} \|\boldsymbol{v}_c - \boldsymbol{v}_s\|^2 + p \sigma^2\\
        &\text{s.t.}
        \begin{aligned}
        & \quad \nabla {\widetilde{f}_{sdf}(\boldsymbol{r})} \boldsymbol{v}_s + \alpha_0(\widetilde{f}_{sdf}(\boldsymbol{r}) - e_{h}) \geq \|\boldsymbol{v}_s\| e_{\nabla h} \\
        & \quad T(\boldsymbol{r})^{\top} \boldsymbol{v}_s  - \Upsilon(\widetilde{f}_{sdf}(\boldsymbol{r}) - e_{h}) \geq \sigma \\
        & \quad \boldsymbol{A}_v \boldsymbol{v}_s \leq \boldsymbol{b}_v 
        \end{aligned}
    \end{aligned}
\end{equation}
where $\sigma$ is a slack variable, introduced to relax the CI constraints in order to ensure the feasibility of the CCBF-SOCP. $p$ is the parameter for $\sigma$. The matrix $\boldsymbol{A}_{v} = \begin{bmatrix}-\boldsymbol{I}_{3\times3} &\boldsymbol{I}_{3\times3} \end{bmatrix}^{\top}$ and vector $\boldsymbol{b}_v = \begin{bmatrix}v_{\max}\boldsymbol{I}_{3\times3} &v_{\max}\boldsymbol{I}_{3\times3} \end{bmatrix}^{\top}$ are used to define the velocity constraints.

\begin{theorem}
\label{thm_cbf_socp}
The solution $\boldsymbol{v}_s = \boldsymbol{0}$ is allways feasible for the proposed CCBF-SOCP problem Eq.(\ref{cbf_socp_CI})
\end{theorem}
\begin{proof}
    We shall verify all of the three constraints in Eq.(\ref{cbf_socp_CI}) and emphasize that the inequalities are always satisfied with solution $\boldsymbol{v}_s = \boldsymbol{0}$.
    \begin{itemize}
        \item For the CBF constraint if $\boldsymbol{v}_s = \boldsymbol{0}$, the following condition is satisfied:
        \begin{equation}
            \alpha_0(\widetilde{f}_{sdf}(\boldsymbol{r}) - e_{h}) \geq 0.
        \end{equation}

        \item The CI constraint is satisfied as:
        \begin{equation}
            -\Upsilon(\widetilde{f}_{sdf}(\boldsymbol{r}) - e_{h}) \geq \sigma
        \end{equation}
        as $\sigma$ is a relaxation parameter, we can always choose appropriate $\sigma$ such that the inequality holds true.
        
        \item The third constraint is satisfied as:
        \begin{equation}
            A_v \boldsymbol{0} \leq b_v
        \end{equation}
        which holds true since $b_v$ is a constant vector.
    \end{itemize}
\end{proof}
\begin{remark}
    \label{remark_ccbf_socp}
    Theorem \ref{thm_cbf_socp} guarantees that the solution $\boldsymbol{v}_s = \boldsymbol{0}$ is always feasible for the CCBF-SOCP problem Eq.(\ref{cbf_socp_CI}). This means that the chaser can always remain at rest, ensuring that it does not collide with the target spacecraft.
\end{remark}

\subsection{Safe Robust Controller Design}
\subsubsection{Disturbance Observer}
Assume that the system disturbance $\boldsymbol{d}$ in Eq.(\ref{dynamics_system}) can be described by the following formula
\begin{equation}
    \label{disturbance_model}
    \dot{\boldsymbol{\xi}}(t) = \boldsymbol{A}\boldsymbol{\xi}(t) , \boldsymbol{d}(t) = \boldsymbol{C} \boldsymbol{\xi}(t)
\end{equation}
where $\xi \in \mathbb{R}^q$, and $\boldsymbol{A}$, $\boldsymbol{C}$ have the appropriate dimensions. 

We can design the nonlinear disturbance observer as:
\begin{equation}
    \label{disturbance_observer}
    \begin{aligned}
        &
        \begin{aligned}
            \dot{\boldsymbol{z}} = 
                &(\boldsymbol{A} - \frac{\phi(\boldsymbol{r})\boldsymbol{C}}{m})\boldsymbol{z} + \boldsymbol{A}r(\boldsymbol{r}) \\ 
                    &- \phi(\boldsymbol{r})(\frac{\boldsymbol{C}\varphi(\boldsymbol{r})}{m} -\boldsymbol{C}_1\boldsymbol{v} - \boldsymbol{C}_2\boldsymbol{r} + \boldsymbol{g} + \frac{\boldsymbol{F}}{m}) \\
        \end{aligned} \\
        &\hat{\boldsymbol{\xi}} = \boldsymbol{z} + \varphi(\boldsymbol{r}), \hat{\boldsymbol{d}} = \boldsymbol{C} \boldsymbol{\xi}
    \end{aligned}
\end{equation}
where the internal state variable $\boldsymbol{z} \in \mathbb{R}^3$ and the nonlinear function $\varphi(\boldsymbol{r})\in \mathbb{R}^3$. $\phi(\boldsymbol{r}) \in  \mathbb{R}^{3\times3}$ is the nonlinear observer gain and designed as $\phi(\boldsymbol{r}) = \frac{\partial \varphi(\boldsymbol{r})}{\partial \boldsymbol{r}}$. 

Let the disturbance estimation error as: $\boldsymbol{e}_d = \boldsymbol{d} - \hat{\boldsymbol{d}}$. By taking the time derivative of $\boldsymbol{e}_d$, one can obtain the error dynamics as follows:
\begin{equation}
    \label{disturbance_observer_error}
    \dot{\boldsymbol{e}}_d = (\boldsymbol{A} - \frac{\phi(\boldsymbol{r}) \boldsymbol{C}}{m} ) \boldsymbol{e}_d.
\end{equation}

\begin{lemma}
\label{lemma_disturbance_observer}
Derived from \cite{chen2004disturbance}, $V_{\boldsymbol{e}_d} = \frac{1}{2} {\boldsymbol{e}_d}^2$ is the Lyapunov function of the estimation error, if $\phi(\boldsymbol{r})$ is selected such that system Eq.(\ref{disturbance_observer_error}) is globally exponentially stable regardless of $\boldsymbol{r}$, there must exist a positive constant $\beta_e$, satisfying the following formula $\dot{V}_{\boldsymbol{e}_d} = -2 \beta_e V_{\boldsymbol{e}_d}$. 
\end{lemma}

\subsubsection{Composite Controller Design}
As we using DO to estimation disturbance, we can design the composite controller ${\boldsymbol{F}_{ref}}$ as follows:
\begin{equation}
    \label{composite_controller}
    \boldsymbol{F}_{ref} = \boldsymbol{u}_{ref}  - \hat{\boldsymbol{d}}
\end{equation}
where $- \hat{\boldsymbol{d}}$ is the disturbance compensation term, and $\boldsymbol{u}_{ref}$ is the state feedback terms designed as:
\begin{equation}
    \begin{aligned}
        \boldsymbol{u}_{ref} = &m(\boldsymbol{C}_1 \boldsymbol{v} + \boldsymbol{C}_2 \boldsymbol{r} - \boldsymbol{g} + \frac{ \partial \boldsymbol{v}_s}{\partial \boldsymbol{r}}\boldsymbol{v} \\ 
            &- \mu_v \frac{\partial V_0}{\partial \boldsymbol{r}} - \frac{\lambda}{2} (\boldsymbol{v}-\boldsymbol{v}_s)),
    \end{aligned}
\end{equation}
with 
\begin{equation}
    \label{labelV0}
    V_0 = \frac{1}{2}\boldsymbol{r}_e^{\top}\boldsymbol{H}_0\boldsymbol{r}_e
\end{equation}
is the Lyapunov function candidate, $\boldsymbol{H}_0$ is a positive definite matrix, $\mu_v > 0$, $\lambda > 0$ are a positive constant. $\frac{\partial \boldsymbol{v}_s}{\partial \boldsymbol{r}}$ is the Jacobian matrix of the reference velocity $\boldsymbol{v}_s$ with respect to the relative position $\boldsymbol{r}$.

\begin{remark}
    In proposed composite controller Eq.(\ref{composite_controller}), the reference controller $\boldsymbol{u}_{ref}$ is designed to stabilize the system Eq.(\ref{relative_orbital_dynamics}) in the absence of disturbance. Due to the existence of disturbance, trajectory of the chaser may not converge to the target position when only $\boldsymbol{u}_{ref}$ is used. To cope with this issue, we design the second disturbance compensation term $- \hat{\boldsymbol{d}}$. 
\end{remark}

\subsubsection{Smooth Safety Filter for Composite Controller}

In Lemma \ref{lemma_disturbance_observer} we have defined a Lyapunov function $V_{\boldsymbol{e}_d} = \frac{1}{2} {\boldsymbol{e}_d}^2$ to guarantee the global exponential stability of the disturbance estimation error $\boldsymbol{e}_d$. And based on the idea of backstepping, the modified CBF $h_1$ considering disturbance estimation is designed as follows: 
\begin{equation}
    \label{backstepping_CBF}
    h_1 = h - \frac{1}{2\mu_h} (\boldsymbol{v} - \boldsymbol{v}_s)^{\top} (\boldsymbol{v} - \boldsymbol{v}_s) - \beta V_{\boldsymbol{e}_d}
\end{equation}
where $\mu_h > 0$, $\beta > 0$ are positive constants.
Define a new safe set $\mathcal{C}_1$ as follows:
\begin{equation}
    \label{backstepping_safe_set}
    \mathcal{C}_1 = \left\{(\boldsymbol{r,v}) \in \mathbb{R}^6| h_1(\boldsymbol{r},\boldsymbol{v}) \geq 0 \right\}.
\end{equation}

Therefore, a smooth safe robust control law
\begin{equation}
    \label{safety_critical_controller}
    \boldsymbol{F} = \boldsymbol{u} - \hat{\boldsymbol{d}}
\end{equation}
is designed to ensure that the system state remains in the safety set $\mathcal{C}_1$. Where $\boldsymbol{u}$ is the smooth modification of the feedback term $\boldsymbol{u}_{ref}$ in $\boldsymbol{F}$ and has the following form:
\begin{equation}
    \begin{aligned}
        \boldsymbol{u} = \boldsymbol{u}_{ref} + \Lambda(a_{h1},b_{h1})\boldsymbol{P}_{h1},
    \end{aligned}
\end{equation}
where $\boldsymbol{P}_{h1}$ is the gradient term, expressed as:
\begin{equation}
    \boldsymbol{P}_{h1} = -\frac{1}{m\mu_h} (\boldsymbol{v} - \boldsymbol{v}_s)^{\top}.
\end{equation}

The scalar coefficient $\Lambda(a_{h1},b_{h1}) \in \mathbb{R}$, determined by $a_{h1}$ and $b_{h1}$, is defined as follows:
\begin{equation}
    \label{SmoothSafetyFiltersParam}
    \begin{aligned}
        a_{h1} = & P_{h1}^{\top} \boldsymbol{u}_{ref} + \nabla \widetilde{f}_{sdf} \boldsymbol{v} - e_{\nabla h} \| \boldsymbol{v} \| \\
                  & + \frac{1}{\mu_h} (\boldsymbol{v} - \boldsymbol{v}_s)^{\top} \left(\boldsymbol{C}_1 \boldsymbol{v} + \boldsymbol{C}_2 \boldsymbol{r} - \boldsymbol{g} + \frac{\partial \boldsymbol{v}_s}{\partial \boldsymbol{r}}\boldsymbol{v}\right) \\
                  & + \beta_c \left(\widetilde{f}_{sdf} - e_{h} - \frac{1}{2\mu_h} (\boldsymbol{v} - \boldsymbol{v}_s)^{\top} (\boldsymbol{v} - \boldsymbol{v}_s)\right) \\
                  & - \frac{\left\| \frac{1}{m\mu_h} (\boldsymbol{v} - \boldsymbol{v}_s)^{\top} \right\|^2}{2\beta(2\beta_e - \beta_c)} \\
        b_{h1} = & \| P_{h1} \|^{2} \\
        \Lambda(a_{h1}, b_{h1}) = 
        &\begin{cases}
            0, & b_{h1} = 0 \\
            \dfrac{1}{\epsilon} \ln\left(1 + e^{-\frac{\epsilon a_{h1}}{b_{h1}}} \right), & b_{h1} \neq 0
        \end{cases}
    \end{aligned}
\end{equation}
where $\epsilon$ and $\beta$, $\beta_e$, $\beta_c$ are positive design parameters.

\begin{remark}
    The proposed controller robustly ensures safety under time-varying disturbances. Unlike methods requiring disturbance estimation error bounds, proposed approach operates without such knowledge. Based on the fact that the velocity $\boldsymbol{v}$ asymptotically converges to its safe velocity $\boldsymbol{v}_s$ and the estimation error $\boldsymbol{e}_d$ converges to zero, we establish that the safe set $\mathcal{C}_1$ asymptotically converges to the nominal safe set $\mathcal{C}$, indicating reduced conservatism in the safe control design. 
\end{remark}

\subsection{Safety and Stability Analysis}
In this section, we give the following theorem about the safety and stability of the proposed control strategy:

\begin{theorem}
    \label{safety_backstepping} 
    Consider the system Eq.(\ref{relative_orbital_dynamics}) with the disturbance model Eq.(\ref{disturbance_model}). The smooth safe robust control law Eq.(\ref{safety_critical_controller}) with DO Eq.(\ref{disturbance_observer}) keep the safe set $\mathcal{C}_1$ defined in Eq.(\ref{backstepping_safe_set}) forward invariant.
\end{theorem}
\begin{proof}
    According to Lemma \ref{lemma_disturbance_observer}, we have that $ \dot{V}_{\boldsymbol{e}_d} = -2 \beta_e V_{\boldsymbol{e}_d}$, where $\beta_e$ is a positive constant. Taking the time derivative of $h_1$ we have
    \begin{equation}
        \begin{aligned}
            \dot{h}_1 
            =\; & (\nabla \widetilde{f}_{sdf} + \boldsymbol{e}_2)\boldsymbol{v} 
            - \frac{1}{\mu_h} (\boldsymbol{v} - \boldsymbol{v}_s)^{\top} 
            \bigg(-\boldsymbol{C}_1\boldsymbol{v} - \boldsymbol{C}_2\boldsymbol{r}  \\
            & + \boldsymbol{g} + \frac{1}{m}\left(\boldsymbol{u} - \hat{\boldsymbol{d}} + \boldsymbol{d} \right) 
            - \frac{\partial \boldsymbol{v}_s}{\partial \boldsymbol{r}}\boldsymbol{v} \bigg) 
            + \beta \beta_e \boldsymbol{e}_d^2 \\[1ex]
            =\; & \Gamma 
            - \frac{1}{m\mu_h} (\boldsymbol{v} - \boldsymbol{v}_s)^{\top} \boldsymbol{u} 
            - \frac{1}{m\mu_h} (\boldsymbol{v} - \boldsymbol{v}_s)^{\top} {\boldsymbol{e}_d} \\
            & + \left(\beta \beta_e - \frac{\beta \beta_c}{2}\right)\boldsymbol{e}_d^2 
            + \frac{\beta \beta_c}{2}\boldsymbol{e}_d^2 \\[1ex]
            =\; & \Gamma 
            - \frac{1}{m\mu_h} (\boldsymbol{v} - \boldsymbol{v}_s)^{\top} \boldsymbol{u} \\
            & + \left\| 
            \sqrt{\beta \beta_e + \frac{\beta \beta_c}{2}}\, \boldsymbol{e}_d 
            + \frac{ - \frac{1}{m\mu_h} (\boldsymbol{v} - \boldsymbol{v}_s)^{\top} }
                { \sqrt{2\beta (2\beta_e - \beta_c)} } 
            \right\|^2 \\
            & + \frac{\beta \beta_c}{2} \boldsymbol{e}_d^2 
            - \frac{ \left\| \frac{1}{m\mu_h} (\boldsymbol{v} - \boldsymbol{v}_s)^{\top} \right\|^2 }
                { 2\beta(2\beta_e - \beta_c) } \\[1ex]
            \geq\; & \Gamma 
            - \frac{1}{m\mu_h} (\boldsymbol{v} - \boldsymbol{v}_s)^{\top} \boldsymbol{u} + \frac{\beta \beta_c}{2} \boldsymbol{e}_d^2 \\
            & - \frac{ \left\| \frac{1}{m\mu_h} (\boldsymbol{v} - \boldsymbol{v}_s)^{\top} \right\|^2 }
                { 2\beta(2\beta_e - \beta_c) }
        \end{aligned}
    \end{equation}

    where $\Gamma = (\nabla \widetilde{f}_{sdf} + \boldsymbol{e}_2)\boldsymbol{v} - \frac{1}{\mu_h} (\boldsymbol{v} - \boldsymbol{v}_s)^{\top} (-\boldsymbol{C}_1 \boldsymbol{v}  -\boldsymbol{C}_2 \boldsymbol{r} + \boldsymbol{g} - \frac{\partial \boldsymbol{v}_s}{\partial \boldsymbol{r}}\boldsymbol{v})$.

    With the CBF constraint $\dot{h}_1 \geq -\beta_c h_1$, we have
    \begin{equation}
        \label{backstepping_CBF_constraint_1}
        \begin{aligned}
            &\Gamma -\frac{1}{m\mu_h} (\boldsymbol{v} - \boldsymbol{v}_s)^{\top} \boldsymbol{u} - \frac{ \|\frac{1}{m\mu_h} (\boldsymbol{v} - \boldsymbol{v}_s)^{\top}\|^2}{2\beta(2\beta_e - \beta_c)} \\\
            &\geq 
            -\beta_c (\widetilde{f}_{sdf} + e_1 - \frac{1}{2\mu_h} (\boldsymbol{v} - \boldsymbol{v}_s)^{\top} (\boldsymbol{v} - \boldsymbol{v}_s))
        \end{aligned}
    \end{equation}

    For any fixed $\boldsymbol{r}$, $\boldsymbol{v}$ and any errors $e_1$ and $\boldsymbol{e}_2$ satisfying Eq.(\ref{estimation_error}), we need the minimum value of the left-hand side greater than the maximum value of the right-hand side to ensure that Eq.(\ref{backstepping_CBF_constraint_1}) still holds, namely:
    \begin{equation}
        \begin{aligned}
            &\min_{|e_1| \leq {e}_h} \left\{ {\Gamma - \frac{\|\frac{1}{m\mu_h} (\boldsymbol{v} - \boldsymbol{v}_s)^{\top}\|^2}{2\beta(2\beta_e - \beta_c)}}\right\} \\ 
            = &\nabla {\widetilde{f}_{sdf}(\boldsymbol{r})}  \boldsymbol{v} -  \|\boldsymbol{v}\| e_{\nabla h} \\
            &- \frac{1}{\mu_h} (\boldsymbol{v} - \boldsymbol{v}_s)^{\top} (-\boldsymbol{C}_1 \boldsymbol{v}  -\boldsymbol{C}_2 \boldsymbol{r} + \boldsymbol{g} - \frac{\partial \boldsymbol{v}_s}{\partial \boldsymbol{r}}(\boldsymbol{v})),
        \end{aligned}
    \end{equation}
    \begin{equation}
        \begin{aligned}
            &\max_{\|\boldsymbol{e}_2\| \leq {e}_{\nabla h}} \left\{ -\beta_c (\widetilde{f}_{sdf} + e_1 - \frac{1}{2\mu_h} (\boldsymbol{v} - \boldsymbol{v}_s)^{\top} (\boldsymbol{v} - \boldsymbol{v}_s))\right\} \\
            =
            & -\beta_c (\widetilde{f}_{sdf} - e_h - \frac{1}{2\mu_h} (\boldsymbol{v} - \boldsymbol{v}_s)^{\top} (\boldsymbol{v} - \boldsymbol{v}_s)).
        \end{aligned}
    \end{equation}

    Therefore, the CBF constraint Eq.(\ref{backstepping_CBF_constraint_1}) can be rewritten as:
    \begin{equation}
        \label{backstepping_CBF_constraint2}
        \begin{aligned}
            &\frac{1}{m\mu_h} (\boldsymbol{v} - \boldsymbol{v}_s)^{\top} {\boldsymbol{u}} \\
            \leq 
            &\nabla \widetilde{f}_{sdf}\boldsymbol{v}  - e_{\nabla h} \| \boldsymbol{v}\| \\
            &+ \frac{1}{\mu_h} (\boldsymbol{v} - \boldsymbol{v}_s)^{\top} (\boldsymbol{C}_1 \boldsymbol{v} + \boldsymbol{C}_2 \boldsymbol{r} - \boldsymbol{g} + \frac{\partial \boldsymbol{v}_s}{\partial \boldsymbol{r}}\boldsymbol{v}) \\
            & + \beta_c (\widetilde{f}_{sdf} - e_{h} - \frac{1}{2\mu_h} (\boldsymbol{v} - \boldsymbol{v}_s)^{\top} (\boldsymbol{v} - \boldsymbol{v}_s))  \\
            &- \frac{\|\frac{1}{m\mu_h} (\boldsymbol{v} - \boldsymbol{v}_s)^{\top}\|^2}{2\beta(2\beta_e - \beta_c)}.
        \end{aligned}
    \end{equation}

    For the QP problem with a single constraint, the analytical solution can be obtained using smooth safety filters \cite{cohen2023characterizing}. Therefore, the form in Eq.(\ref{SmoothSafetyFiltersParam}) can be obtained through Eq.(\ref{backstepping_CBF_constraint2}). 
\end{proof}

However, for the dynamical system described by Eq.(\ref{relative_orbital_dynamics}), it is as anticipated that the composite controller $\boldsymbol{F}$, derived after the application of the safety velocity filter and safety control modifications, cannot fully guarantee system convergence. This stems from the potential conflict between safety and stability requirements. For instance, when the reference velocity $\boldsymbol{v}_s$ drives the chaser towards an unsafe region, the safety velocity filter yielding $\boldsymbol{v}_s$ and the safe robust controller $\boldsymbol{F}$ will prioritize system safety to avoid obstacles, inevitably compromising stability and convergence rate.

On the other hand, the safety velocity filter and safe robust controller can be interpreted as modifications to the reference velocity and reference controller, i.e., $\boldsymbol{v}_s = \boldsymbol{v}_c + \delta_v$ and $\boldsymbol{F} = \boldsymbol{F}_{ref} + \delta_{F}$. When the trajectory of the system approaches the boundary of the safe region, $\delta_{v}$ and $\delta_{F}$ are invoked to ensure safety. Conversely, when the trajectory is sufficiently far from the boundary, $\delta_{v}$ and $\delta_{F}$  asymptotically converge to zero to preserve stability and convergence. Therefore, for the purpose of stability analysis, we can assume that the reference velocity and reference control are safe, i.e., $\boldsymbol{v}_s = \boldsymbol{v}_c$ and $\boldsymbol{F} = \boldsymbol{F}_{ref}$ thereby ensuring system convergence.

\begin{theorem}
\label{stability_backstepping}
The composite controller Eq.(\ref{composite_controller}) will stabilize the system Eq.(\ref{relative_orbital_dynamics}) to the desired position $\boldsymbol{r}_d$ in the presence of disturbance $\boldsymbol{d}$ when the the safety velocity $\boldsymbol{v}_s$ is ultimately equal to $\boldsymbol{v}_c$
\end{theorem}

\begin{proof}
    As Definition \ref{ISS_CLF}, Let $V_1$ denote an ISS-CLF for system Eq.(\ref{relative_orbital_dynamics}). Any nominal control input $\boldsymbol{u}_{ref}$ satisfying the constraint
    \begin{equation}
        \label{input_constraint}
        L_fV_1 + L_gV_1 \boldsymbol{u}_{ref} \leq -\sigma_V V_1
    \end{equation}
    can guarantee the input-to-state stability of system Eq.(\ref{relative_orbital_dynamics}) in the presence of estimation error $\boldsymbol{e}_d$, which means that the composite controller Eq.(\ref{composite_controller}) can exponentially stabilize the nonlinear system since the estimation error $\boldsymbol{e}_d$ is exponentially stable.
    The $V_1$ is defined as follows:
    \begin{equation}
        V_1 = V_0 + \frac{1}{2\mu_v} (\boldsymbol{v} - \boldsymbol{v}_s)^{\top} (\boldsymbol{v} - \boldsymbol{v}_s).
    \end{equation}
    With $V_0$ defined in Eq.(\ref{labelV0}), $\dot{V_1}$ can be expressed as:
    \begin{equation}
        \begin{aligned}
            &\dot{V_0} + \frac{1}{\mu_v} (\boldsymbol{v} - \boldsymbol{v}_s)^{\top} (-\boldsymbol{C}_1 \boldsymbol{v}  -\boldsymbol{C}_2 \boldsymbol{r} + \boldsymbol{g} + \frac{1}{m} \boldsymbol{u}_{ref} - \frac{\partial \boldsymbol{v}_s }{\partial \boldsymbol{r}} \boldsymbol{v}) \\
            &= \frac{\partial V_0}{\partial \boldsymbol{r}}\boldsymbol{v} + \frac{1}{\mu_v}( -\mu_v \frac{\partial V_0}{\partial \boldsymbol{r}} - \frac{\lambda}{2}(\boldsymbol{v} - \boldsymbol{v}_s)) \\
            &= \frac{\partial V_0}{\partial \boldsymbol{r}}\boldsymbol{v}_s - \frac{\lambda}{2 \mu_v}(\boldsymbol{v} - \boldsymbol{v}_s)^{\top} (\boldsymbol{v} - \boldsymbol{v}_s)\\
            &= (\boldsymbol{r} - \boldsymbol{r}_d )^{\top} \boldsymbol{H}_0 \boldsymbol{v}_s - \frac{\lambda}{2 \mu_v}(\boldsymbol{v} - \boldsymbol{v}_s)^{\top} (\boldsymbol{v} - \boldsymbol{v}_s).
        \end{aligned}
        \label{input_constraint2}
    \end{equation}
    Since $\boldsymbol{v}_s$ is ultimately equal to $\boldsymbol{v}_c$, taking $\boldsymbol{v}_c$ into Eq.(\ref{input_constraint2}), we obtain:
    \begin{equation}
        \begin{aligned}
            \dot{V_1} &= -v_{\max} \boldsymbol{r}_e^{\top} \boldsymbol{H}_0 \boldsymbol{r}_e - \frac{\lambda}{2 \mu_v}(\boldsymbol{v} - \boldsymbol{v}_s)^{\top} (\boldsymbol{v} - \boldsymbol{v}_s) \\
            &\leq -\sigma_1 V_0 - \sigma_2 \frac{1}{2\mu_v} (\boldsymbol{v} - \boldsymbol{v}_s)^{\top} (\boldsymbol{v} - \boldsymbol{v}_s) \\
            &\leq -\min(\sigma_1,\sigma_2)V_1 \\
            & = -\sigma_V V_1.
        \end{aligned}
    \end{equation}
    
    Since $\sigma_1$ and $\sigma_2$ are positive constants. we obtain $\dot{V_1} \leq -\sigma_V V_1$, from which it can be concluded that the composite controller given in Eq.(\ref{composite_controller}) can guarantee the input-to-state stability of system Eq.(\ref{relative_orbital_dynamics}) in the presence of disturbance $\boldsymbol{d}$. 
\end{proof}

\section{SIMULATION RESULTS AND ANALYSIS}
\label{simulation}
In this section, the numerical simulation is conducted to illustrate the effectiveness of the proposed control strategy.

\subsection{Simulation Settings}
\subsubsection{Network Architecture and Training Implementation}
For representing target spacecraft model we used level sets of MLP with 10 layers, each contains 512 hidden units. Set loss parameters $\kappa = 2$, $\eta = 0.1$. The training data set samples 50k points, using ADAM optimizer \cite{adam2014method} for 10k iterations with a adjust learning rate strategy with initial learning rate 0.005, and decay rate 0.5 every 2000 iterations. 

\subsubsection{Simulation Parameters}
The parameters of the chaser are listed as follows: $m = 20$kg, $v_{max} = 0.1$m/s, $\boldsymbol{F}_{max} = 0.1$N.
The control strategy parameters are selected as: $p = 1$, $\alpha_0(\boldsymbol{r}) = 0.08\boldsymbol{r}$, $\Omega = \begin{bmatrix}    0 & 0 & 1 \\    -1 & 0 & 0 \\    0 & 0 & 0\end{bmatrix}$, $\Upsilon(x) = 0.1 - x$ and DO parameters are set as:$\boldsymbol{A} = I_{3\times3}$, $\boldsymbol{C} = I_{3\times3}$, $\varphi  = [1,1,1]$, $\phi  = I_{3\times3}$, $\boldsymbol{L}  = 50I_{3\times3}$, and $\beta_c = 10$, $\beta_e = 0.1$, $\beta = 1$, $\mu_h = 2$, $\mu_v = 0.0001$, $\lambda = 15$. The external disturbance is set as $\boldsymbol{d} = 0.01 \left[\sin(0.02t), \cos(0.02t), \sin(0.01t) \right] ^{\top}$N, which are large enough compared with the real external disturbance in space. The SOCP problem is solved by using the CVXPY \cite{diamond2016cvxpy,agrawal2018rewriting} in our simulation and the control frequency is set as 1 Hz, with Intel Core i7-12700F, NVIDIA GeForce RTX 2080 SUPER, and 64GB RAM.

The proposed methods are validated through four simulation experiments: (1) comparison of the proposed and traditional target spacecraft modeling approaches; (2) demonstration of CI constraint in avoiding local minima by comparing trajectories with and without it; (3) evaluation of the controller's safety performance under external disturbances; and (4) overall validation of the control framework via Monte Carlo simulations

\subsection{SDF Estimation Results}
To demonstrate the generality and effectiveness of the proposed spacecraft representation method, this section evaluates its applicability using the International Space Station (ISS), the Chinese Space Station (CSS), and the commercial communication satellite InterSat-30 as representative spacecraft. Their respective SDF models are trained and evaluated to assess the applicability of the method. The 3D models of the selected target spacecraft are shown in Fig.(\ref{tISS}), Fig.(\ref{tCSS}), and Fig.(\ref{tintersat}), respectively.

\begin{figure*}[h!]
    \centering
    \subfigure[ISS]{
        \centering
        \includegraphics[scale=0.75]{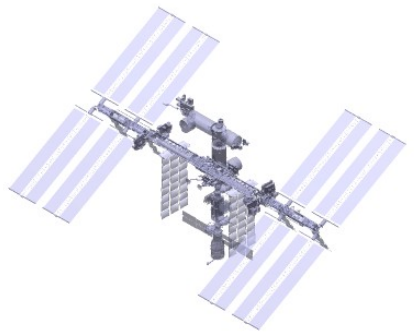}
        \label{tISS}}
    \centering
    \subfigure[CSS]{
        \centering
        \includegraphics[scale=0.75]{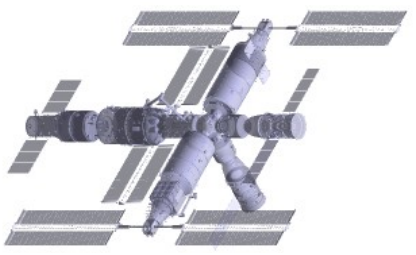}
        \label{tCSS}}
    \centering
    \subfigure[Intersat-30]{
        \centering
        \includegraphics[scale=0.75]{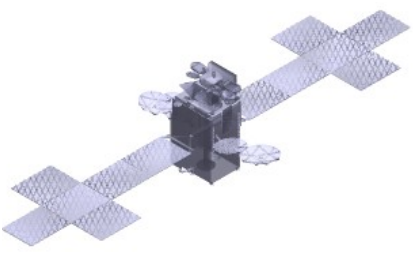}
        \label{tintersat}}
    \caption{Target 3D model}
    \label{target_shape}
\end{figure*} 

We measure SDF approximation error and under-approximation error as:
\begin{equation}
    \begin{aligned}
        &\varepsilon = \frac{1}{N} \sum_{i=1}^{N} |\widetilde{f}_{sdf}(\boldsymbol{y}_i)|\\
        &\varepsilon_+ = \frac{1}{|\mathcal{I}_+|} \sum_{i \in \mathcal{I}_+} \widetilde{f}_{\mathrm{sdf}}(\boldsymbol{y}_i) \\
        &\mathcal{I}_+ = \left\{ i \in \{1, \dots, N\} \mid \widetilde{f}_{\mathrm{sdf}}(\boldsymbol{y}_i) > 0 \right\}
    \end{aligned}
\end{equation}
where $N = 100000$ points uniformly sampled on the surface of the target. 
Table.(\ref{ESTIMATIONERROR}) shows the approximate effect evaluation of the three selected target spacecraft.

\begin{table}[hbt!]
\caption{SDF ESTIMATION ERROR}
\label{ESTIMATIONERROR}
\centering
\begin{tabular*}{20pc}{@{\extracolsep{\fill}}lcccc}
\hline
Targets & $\varepsilon$  & $\varepsilon_+$  & $e_h$ \\
\hline
ISS         & 0.0011 & 0.0011 & 0.019 \\
CSS         & 0.002  & 0.0013 & 0.012 \\
Intersat-30 & 0.0015 & 0.0014 & 0.031 \\
\hline
\end{tabular*}
\end{table}

Using the $Marching Cubes$ meshing algorithm \cite{Marchingcubes} to reconstruct the surface of the target. The neural SDF approximate surface of the selected target spacecraft is shown in Fig.(\ref{ISS}), Fig.(\ref{CSS}), and Fig.(\ref{intersat}).

\begin{figure*}[h]
    \centering
    \subfigure[ISS]{
        \centering
        \includegraphics[scale=0.75]{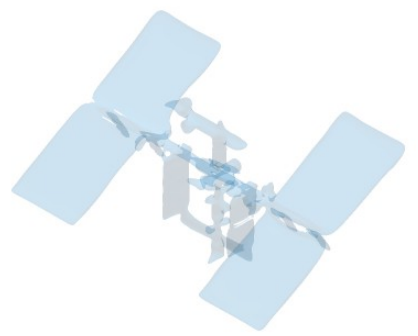}
        \label{ISS}}
    \centering
    \subfigure[CSS]{
        \centering
        \includegraphics[scale=0.75]{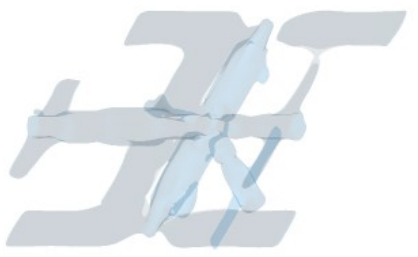}
        \label{CSS}}
    \centering
    \subfigure[Intersat-30]{
        \centering
        \includegraphics[scale=0.75]{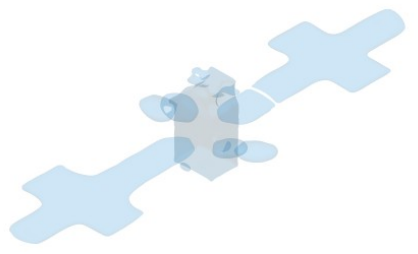}
        \label{intersat}}
    \caption{Estimated SDF}
    \label{estimated SDF}
\end{figure*}


\begin{figure*}[h!]
    \centering
    \includegraphics[scale=0.30]{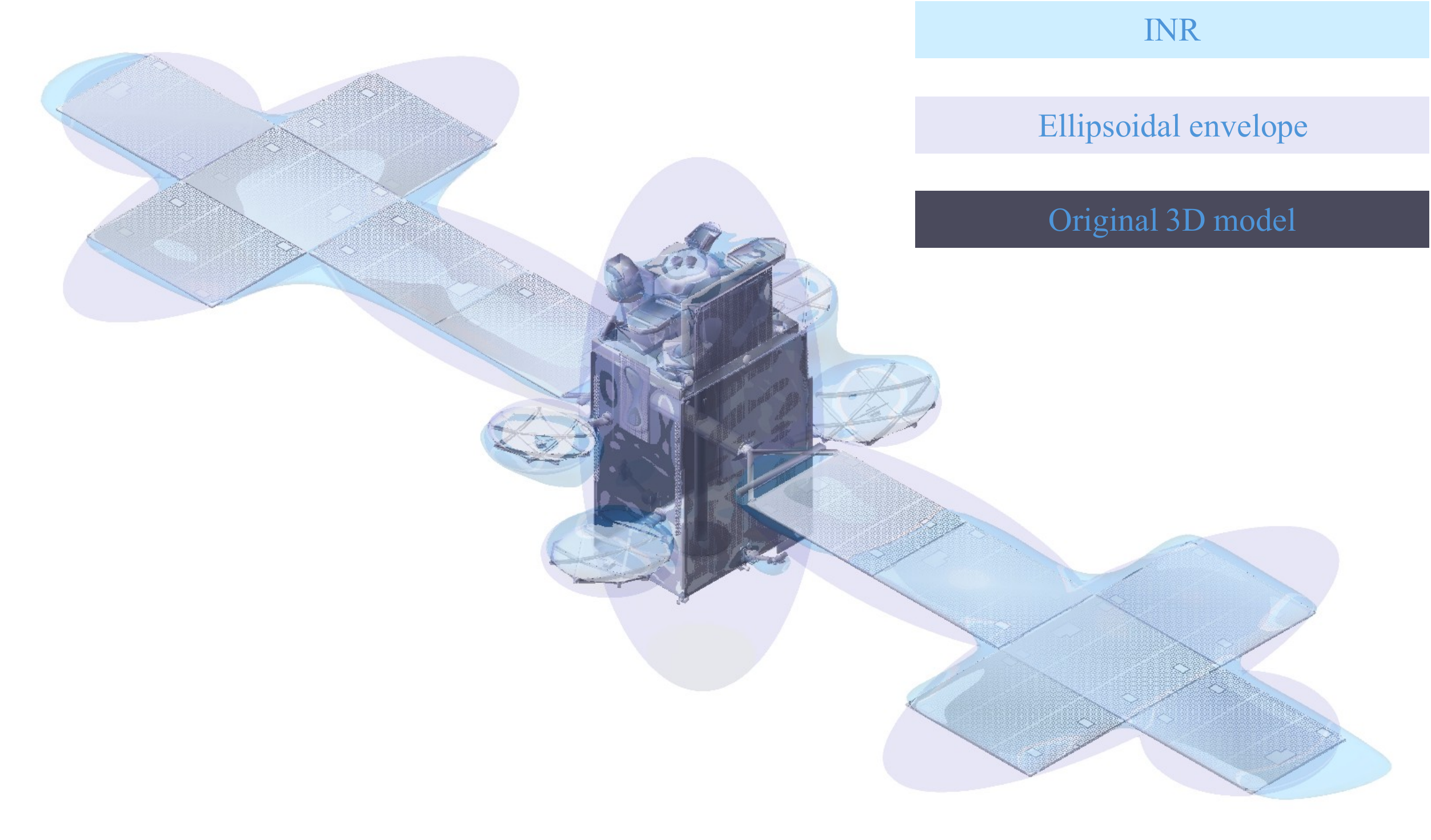}
    \caption{Target representation comparison}
    \label{Estimation_fig_intersat30}
\end{figure*}

As shown in Fig.(\ref{Estimation_fig_intersat30}), the proposed target spacecraft representation is compared with the conventional ellipsoidal envelope method \cite{wang2025safe}. The volumes of the INR and the ellipsoidal approximation are calculated as $91.4\rm{m^3}$ and $270.6\rm{m^3}$, respectively, demonstrating that our method yields a significantly tighter and more precise geometric envelope than the traditional approach.

\subsection{Controller Simulation Results}
In this section, the target spacecraft is selected as the commercial communication satellite InterSat-30 to verify the effectiveness of the proposed controller.
\subsubsection{case 1: Local Minimum Confirmation}
In this case, the goal position is set as $\boldsymbol{r}_d = [0, -10, -4]^{\top} \rm{m}$ and the initial position is set as $\boldsymbol{r}_0 = [0, 10, 6]^{\top} \rm{m}$.

Fig.(\ref{local_minimum_traj}) shows the simulated spacecraft trajectories without and with the proposed CI strategy, respectively. As illustrated in Fig.(\ref{local_minimum_fig1}), the trajectory without CI becomes trapped in a local minimum and fails to reach the target position. In contrast, Fig.(\ref{local_minimum_fig2}) demonstrates that, with CI enabled, the spacecraft successfully reaches the desired target. 
Fig.(\ref{local_minimum_cbf}) presents the CBF values for both scenarios. In both cases, the CBF remains positive, indicating that the spacecraft stays within the safe region. However, without CI, the CBF value gradually approaches zero, implying that the spacecraft remains on the boundary of an obstacle and is trapped in a local minimum. 
Fig.(\ref{local_minimum_re}) shows the position tracking error under both conditions. With CI, the error converges to zero, confirming that the spacecraft successfully reaches the target location successfully. 
Fig.(\ref{local_minimum_vs}) illustrates the safe velocity with and without the CI constraint, while Fig.(\ref{local_minimum_v}) presents the corresponding actual velocity profiles. The safe velocity generated by the proposed CCBF-SOCP method consistently satisfies the velocity constraint. Likewise, the actual velocity also remains within the allowable bounds.

\begin{figure}[hbt!]
    \centering
    \subfigure[without CI]{
        \centering
        \includegraphics[scale=0.3]{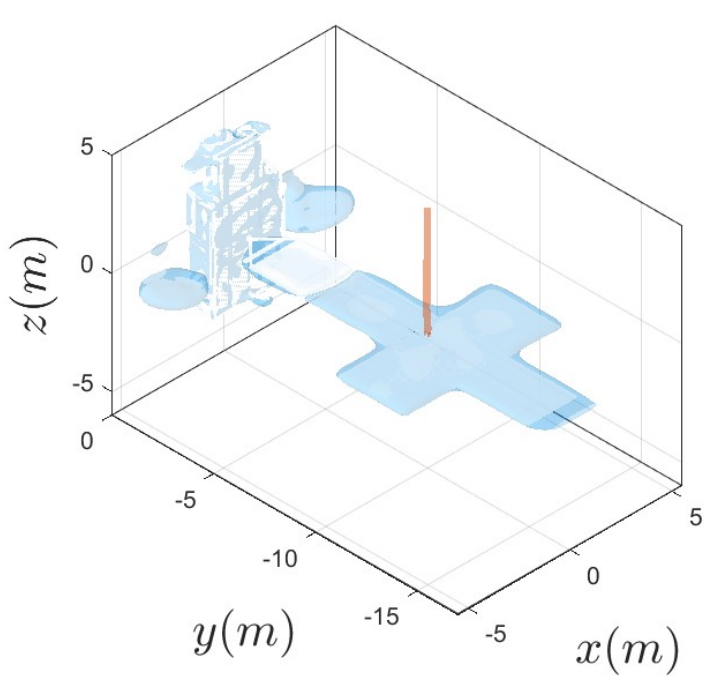}
        \label{local_minimum_fig1}}
    \centering
    \subfigure[with CI]{
        \centering
        \includegraphics[scale=0.3]{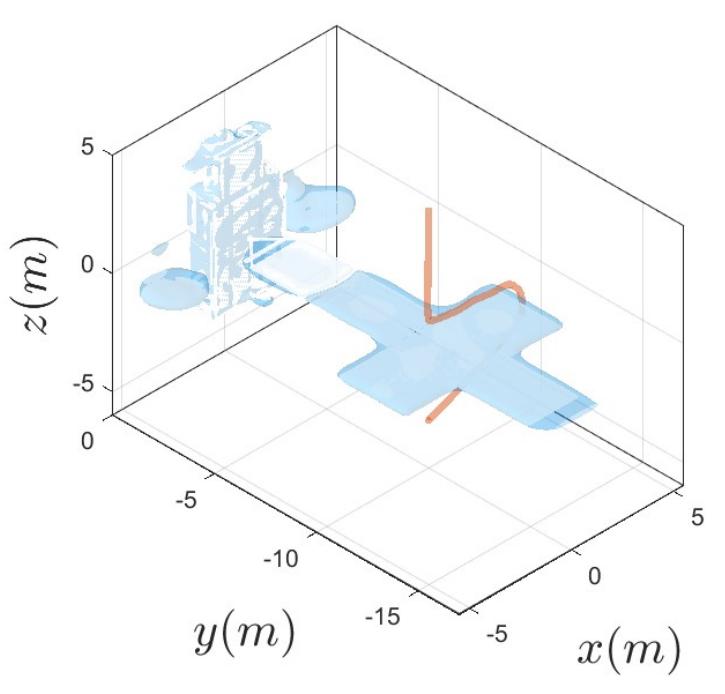}
        \label{local_minimum_fig2}}
    \centering
    \caption{Simulation trajectory with CI or without CI}
    \label{local_minimum_traj}
\end{figure}

\begin{figure}[hbt!]
    \centering
    \includegraphics[scale=0.42]{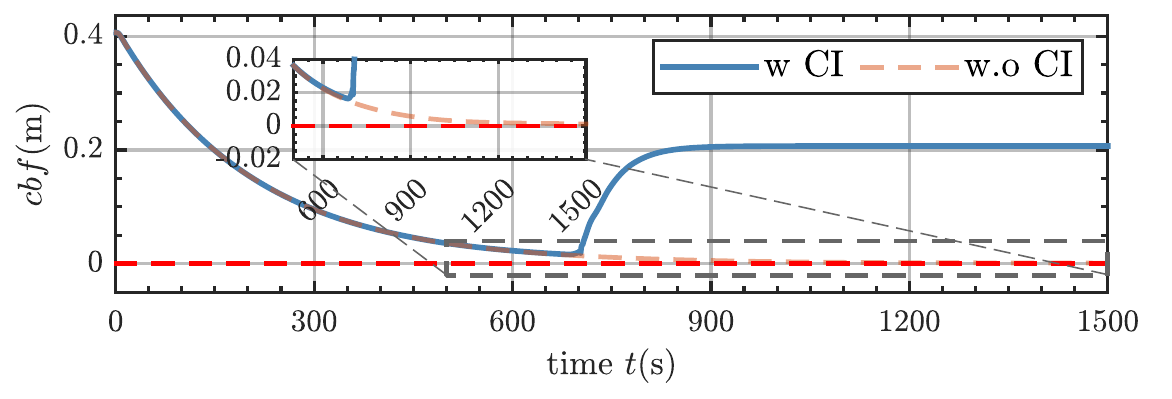}
    \caption{CBF value with or without CI}
    \label{local_minimum_cbf}
\end{figure}

\begin{figure}[hbt!]
    \centering
    \includegraphics[scale=0.42]{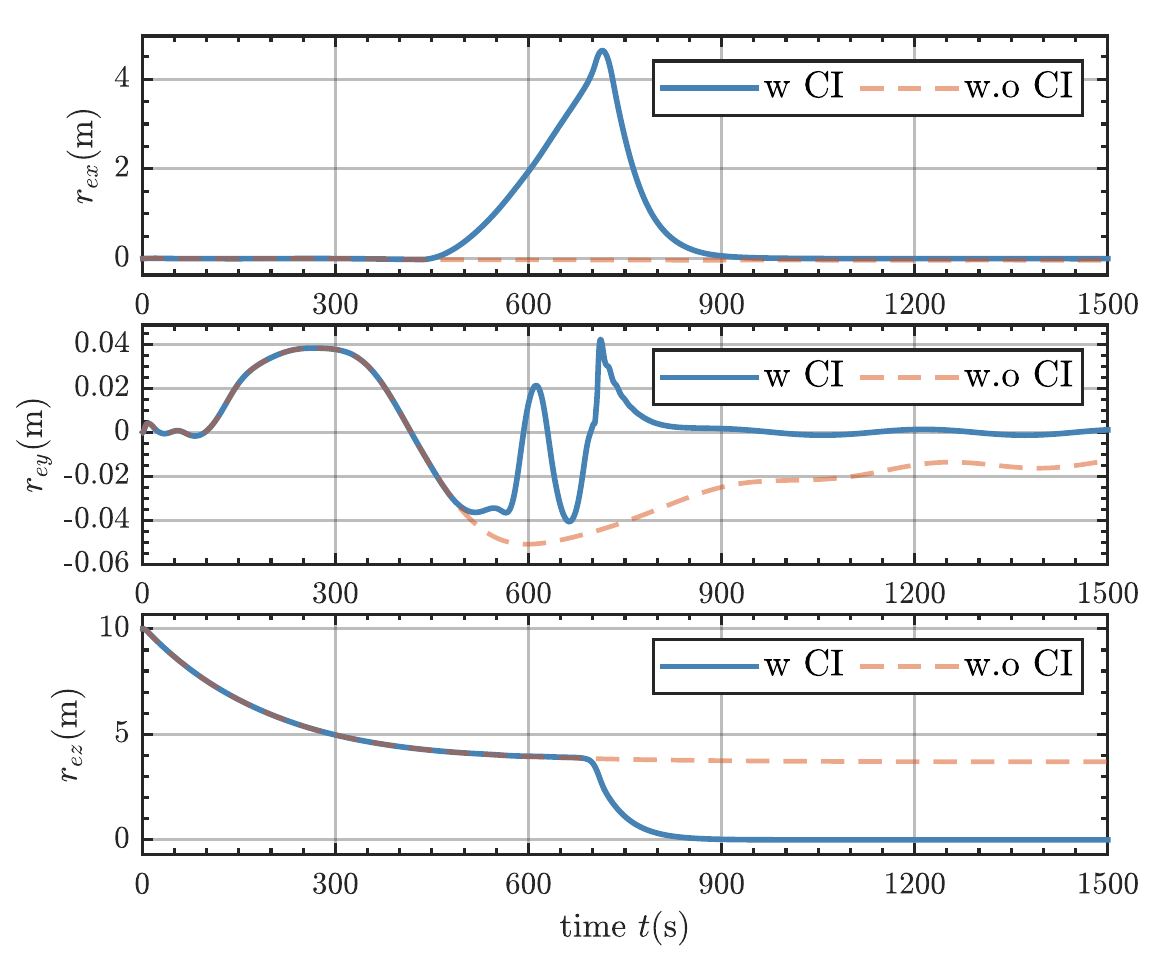}
    \caption{Position error with or without CI}
    \label{local_minimum_re}
\end{figure}

\begin{figure}[hbt!]
    \centering
    \includegraphics[scale=0.42]{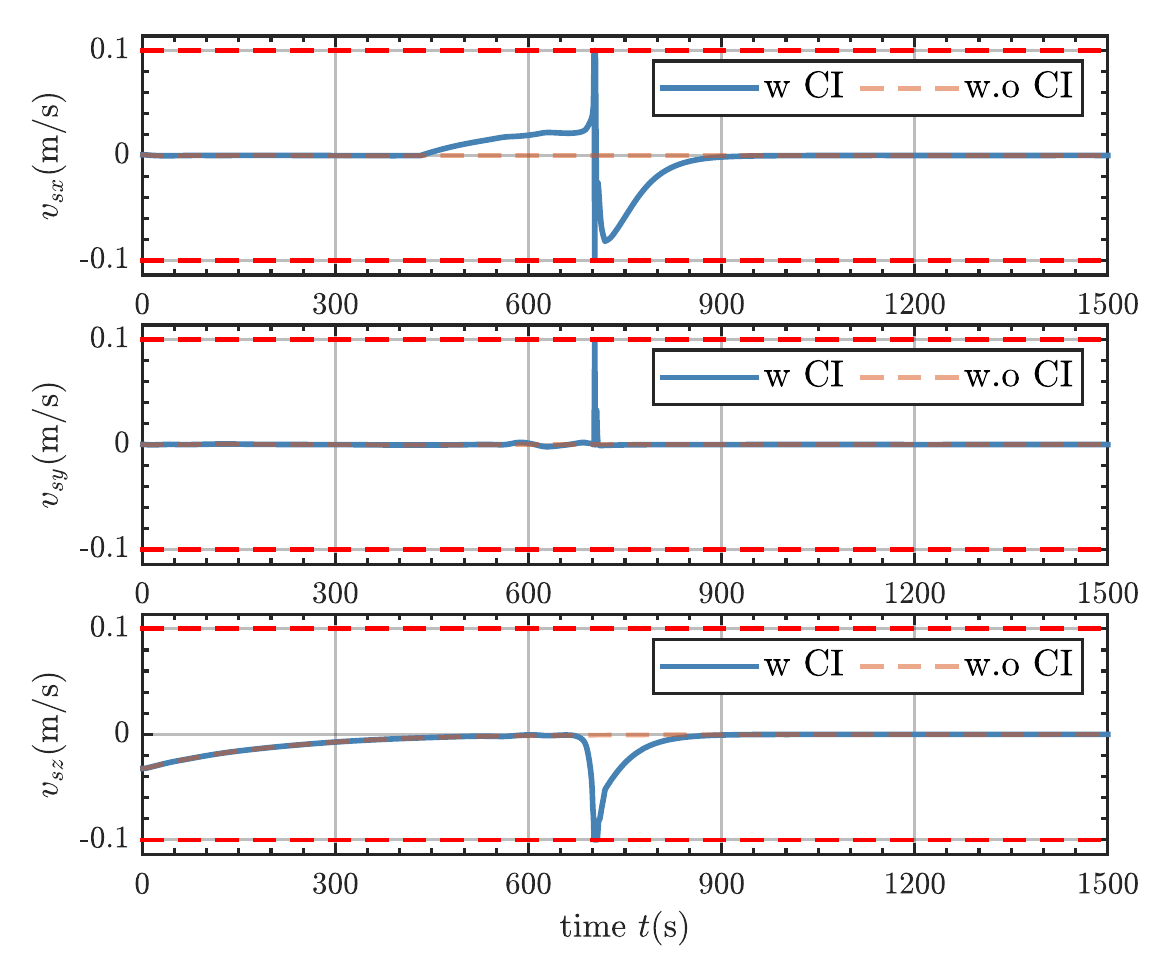}
    \caption{Safe velocity with or without CI}
    \label{local_minimum_vs}
\end{figure}

\begin{figure}[hbt!]
    \centering
    \includegraphics[scale=0.42]{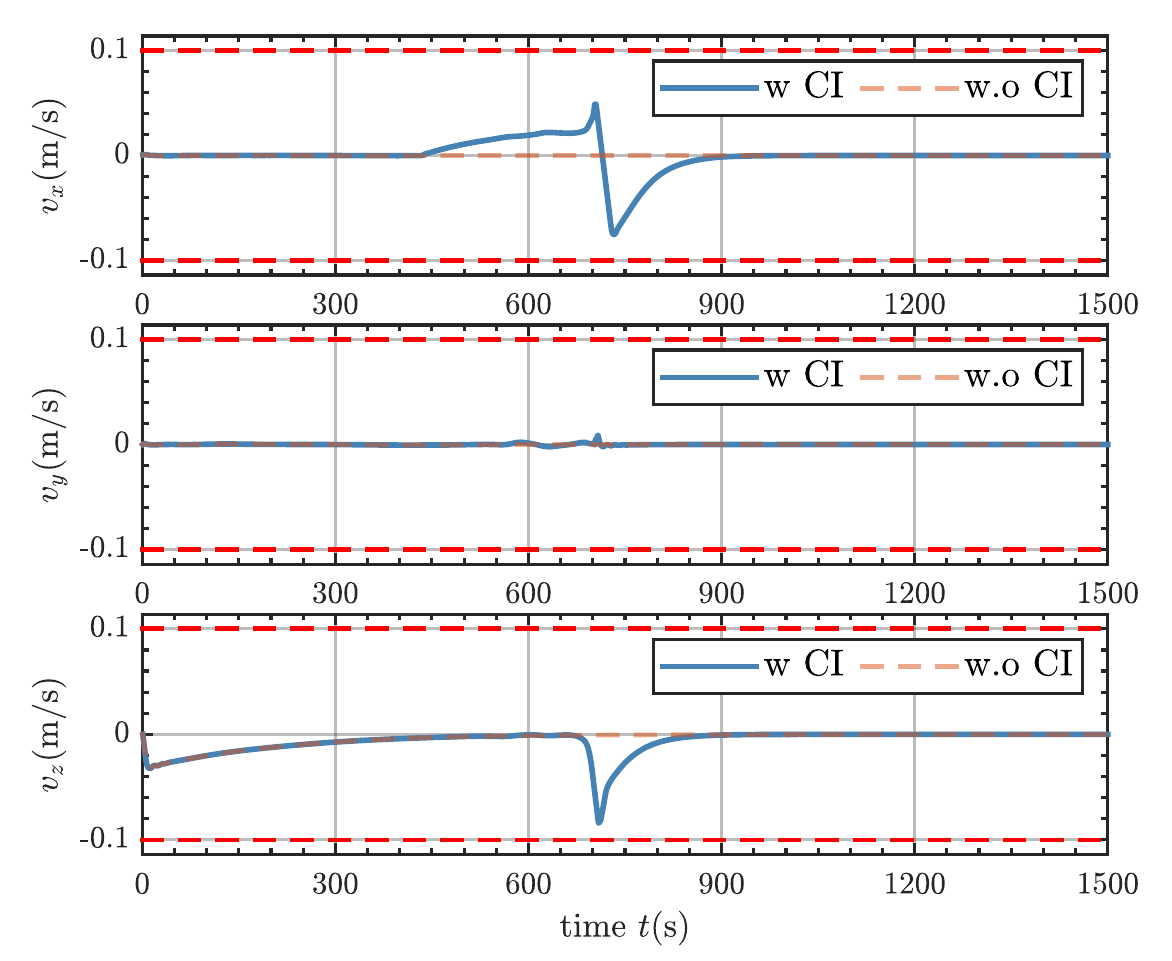}
    \caption{Velocity with or without CI}
    \label{local_minimum_v}
\end{figure}
\subsubsection{case 2: Robustness Confirmation}
This section compares the controller effects with and without DO to illustrate the robustness of DO to the controller under disturbance. 
Fig.(\ref{DOre}) shows the position error with and without DO. The position error with DO converges to zero, indicating that the spacecraft can successfully reach the target position. Fig.(\ref{DO_cbf}) shows the CBF value with and without DO. The CBF value with DO is always positive, indicating that the spacecraft is always in a safe state with external disturbances. Fig.(\ref{DO_clf}) and Fig.(\ref{DO_Fo}) show the CLF value and control input value with and without DO, respectively.

Fig.(\ref{runtime}) presents the computation time per control step of the proposed method, which includes the time required to solve the SCOP problem and to perform inference on the neural SDF model. It can be observed that the computation time remains below 10 ms throughout the simulation. This indicates that, relative to the control period, the solution time is negligible. Therefore, the proposed approach imposes relatively low computational demands and can be readily supported by the onboard processing capabilities of existing spacecraft.

\begin{figure}[hbt!]
    \centering
    \includegraphics[scale=0.42]{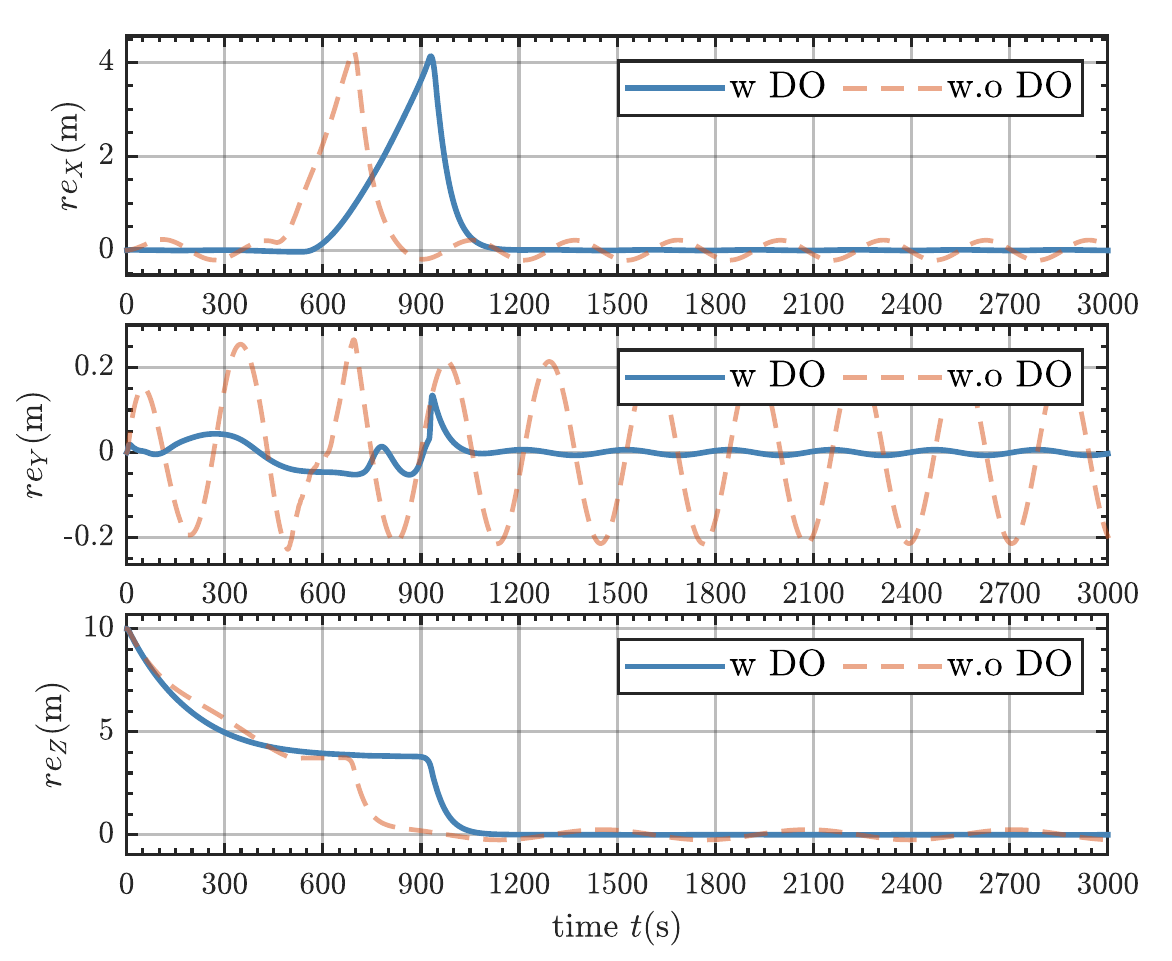}
    \caption{Position error with or without DO}
    \label{DOre}
\end{figure}

\begin{figure}[hbt!]
    \centering
    \includegraphics[scale=0.42]{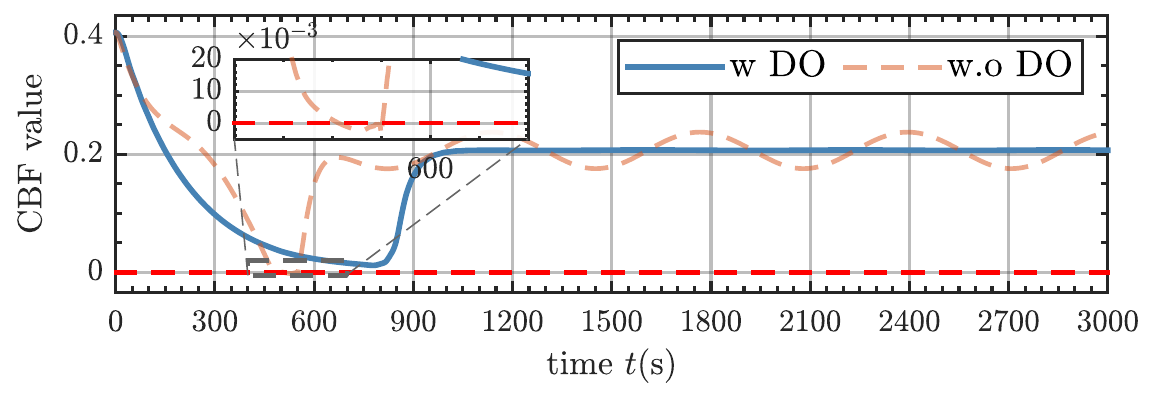}
    \caption{CBF value with or without DO}
    \label{DO_cbf}
\end{figure}

\begin{figure}[hbt!]
    \centering
    \includegraphics[scale=0.42]{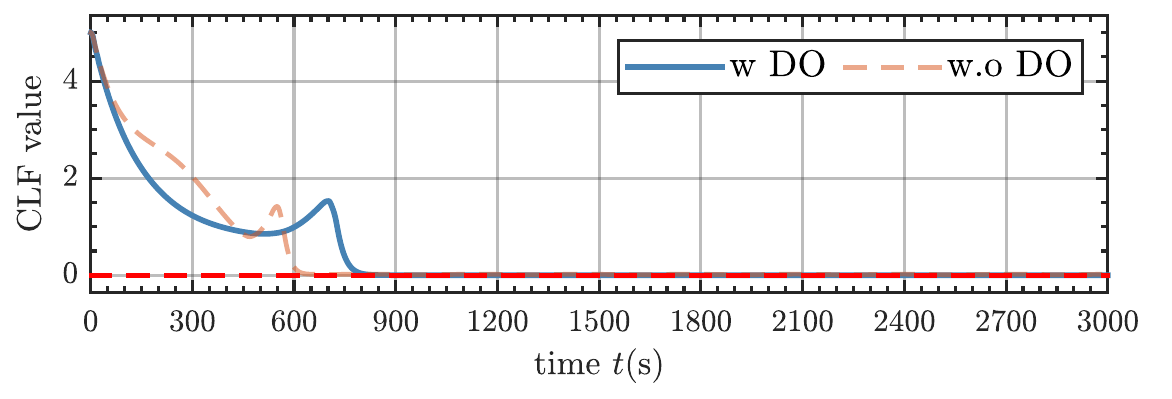}
    \caption{CLF value with or without DO}
    \label{DO_clf}
\end{figure}

\begin{figure}[hbt!]
    \centering
    \includegraphics[scale=0.42]{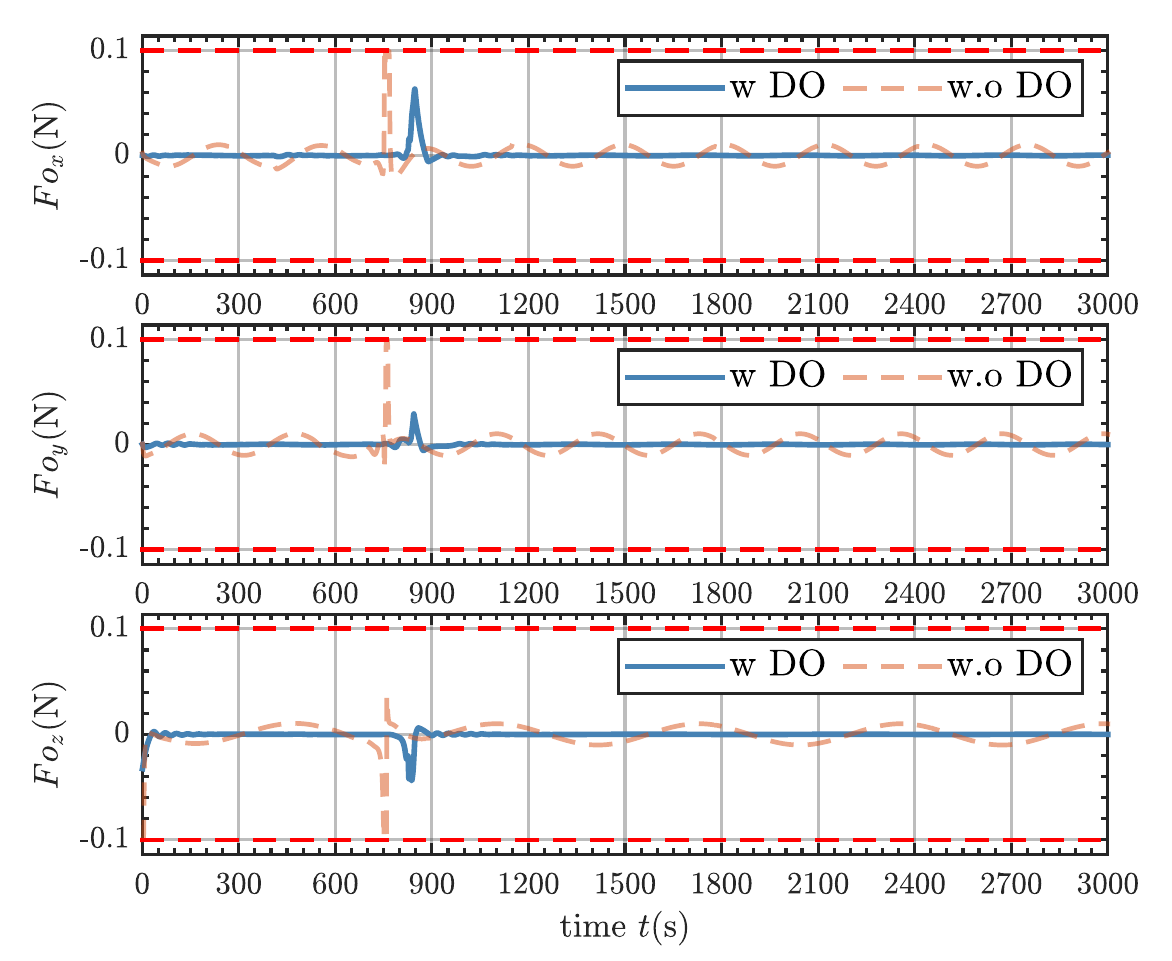}
    \caption{Control input with or without DO}
    \label{DO_Fo}
\end{figure}

\begin{figure}[hbt!]
    \centering
    \includegraphics[scale=0.08]{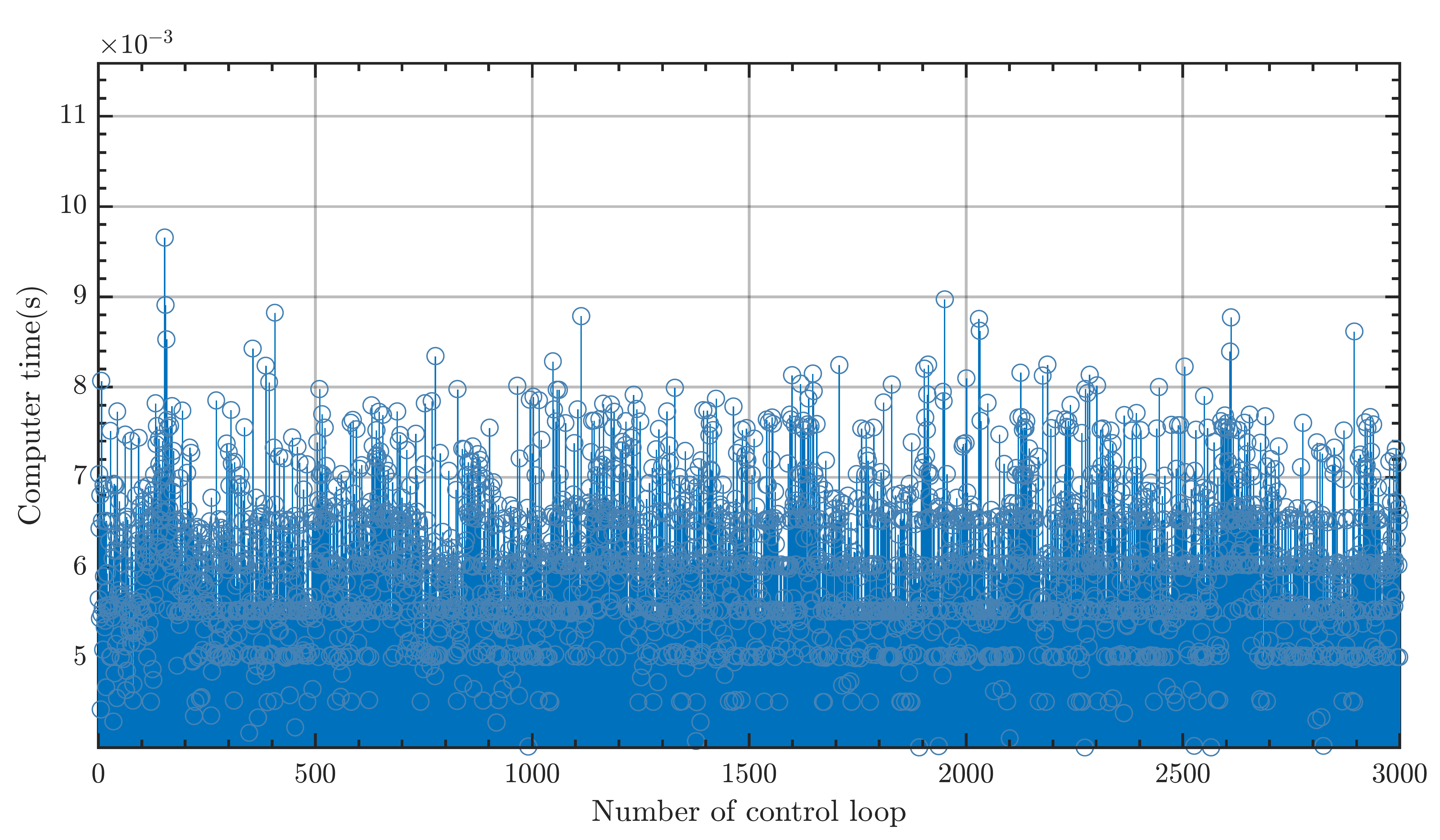}
    \caption{Computing time of the proposed controller}
    \label{runtime}
\end{figure}

\subsubsection{case 3: Monte Carlo Simulation}
This section verifies the safety and stability of the proposed safe control framework, as well as its effectiveness, through Monte Carlo simulations.

A total of 1000 simulations were conducted with randomly sampled initial positions. The initial states were sampled as $\boldsymbol{r}_s = [x, y, z]$. For the Intersat-30 scenario, the sampling ranges were $x \in [-15, 15]\rm{m}$ , $y \in [-5, 5]\rm{m}$, and $z = -4\rm{m}$. The goal position was set to $\boldsymbol{r}_d = [0, -10, 6]^{\top}\rm{m}$, and the simulation time was set to 3000 seconds. For the ISS scenario, the initial position was sampled from $x \in [-60, 60]\rm{m}$, $y \in [-40, 40]\rm{m}$, $z = -30\rm{m}$, and the goal position was set to $\boldsymbol{r}_d = [0, -10, 6]^{\top}\rm{m}$. with the simulation time set to 4000 seconds. 

Fig.(\ref{cbf_w_cir}) and Fig.(\ref{cbf_w_cir_iss}) illustrate the CBF values obtained from all simulations. The CBF values remain strictly positive throughout, indicating that all trajectories satisfy the safety constraints. Fig.(\ref{clf_w_cir}) and Fig.(\ref{clf_w_cir_iss}) show the corresponding CLF values. All CLF values converge to zero, demonstrating the asymptotic stability of the system across all runs. Fig.(\ref{pose_error_w_cir}) and Fig.(\ref{pose_error_w_cir_iss}) present the position tracking errors. All errors converge to zero, confirming that each time the chaser successfully reaches the designated target position. Fig.(\ref{velocity_w_cir}) and Fig.(\ref{velocity_w_cir_iss}) present the actual velocity of the chaser. It can be observed that the velocity may violate the velocity constraint defined in Eq.(\ref{relative_velocity_constraint}). This occurs because the velocity constraint is implemented as a soft constraint, as discussed in Remark \ref{remark_velocity_constraint}. 

\begin{figure}[hbt!]
    \centering
    \subfigure[CBF value]{
        \centering
        \includegraphics[scale=0.3]{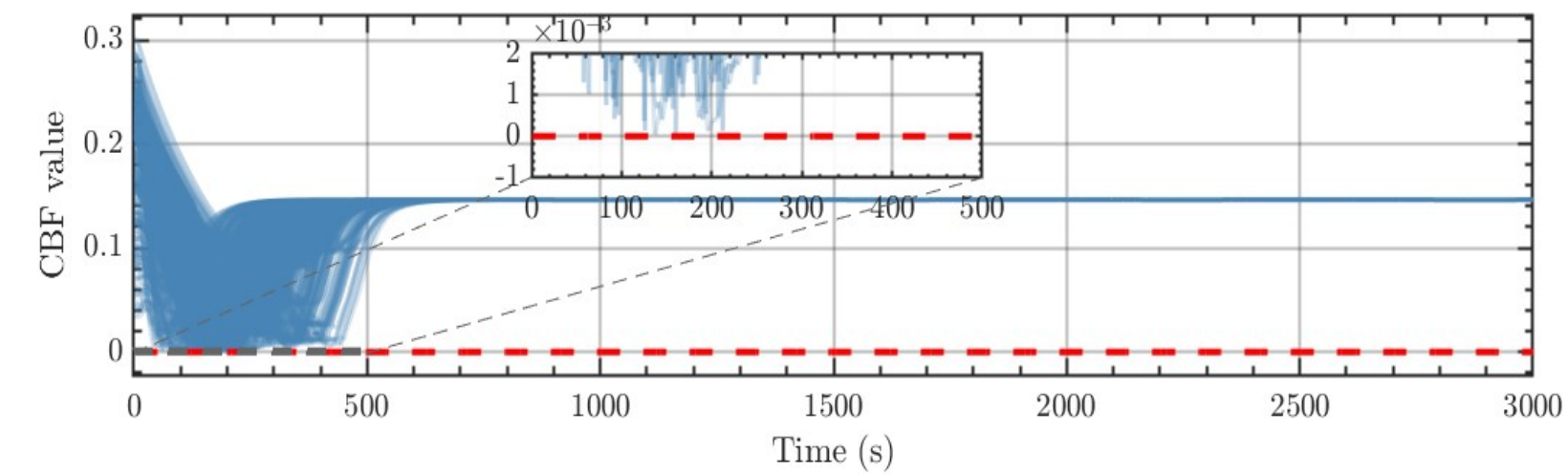}
        \label{cbf_w_cir}}
    \centering
    \subfigure[CLF value]{
        \centering
        \includegraphics[scale=0.3]{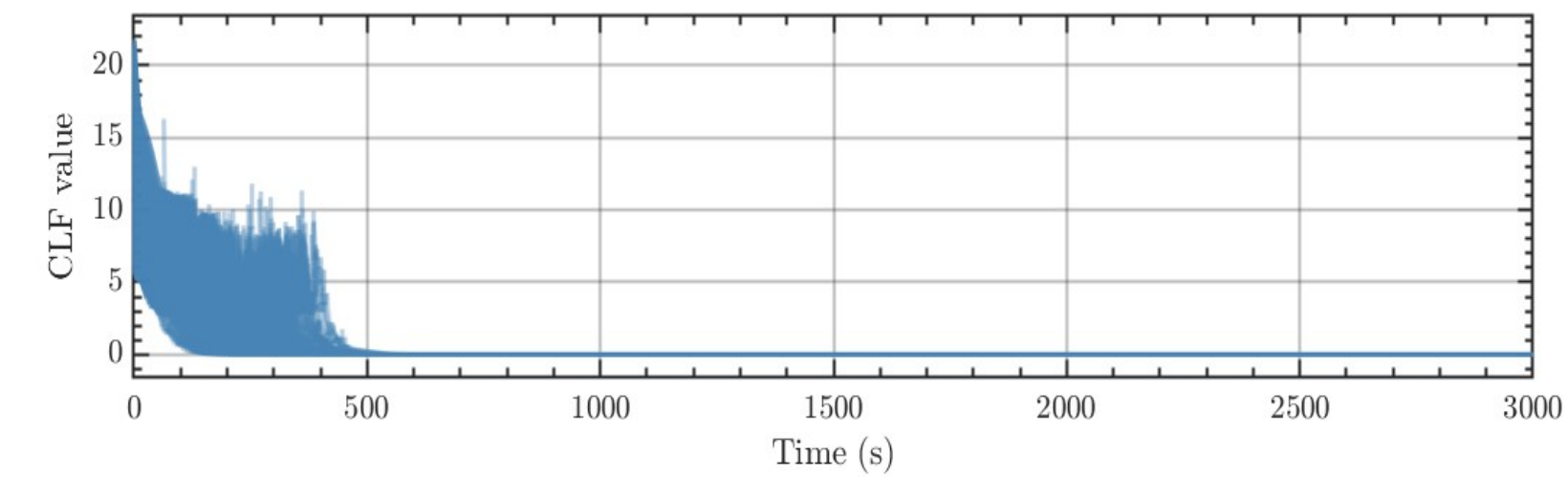}
        \label{clf_w_cir}}
    \centering
    \subfigure[Position error]{
        \centering
        \includegraphics[scale=0.3]{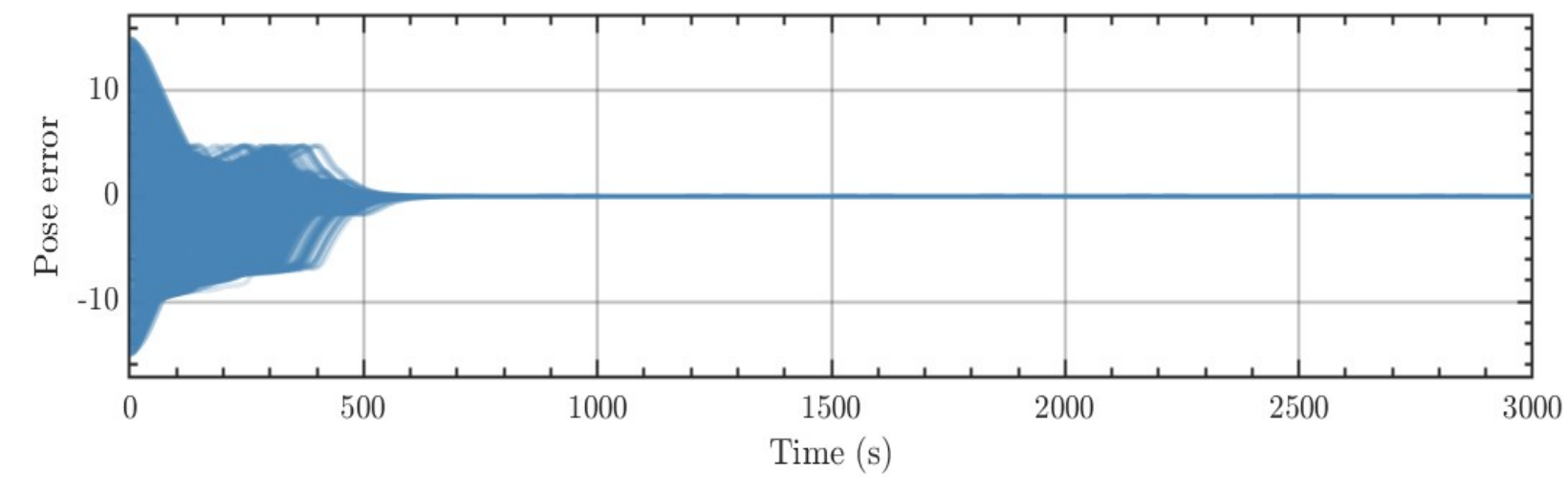}
        \label{pose_error_w_cir}}
    \subfigure[Velocity]{
        \centering
        \includegraphics[scale=0.3]{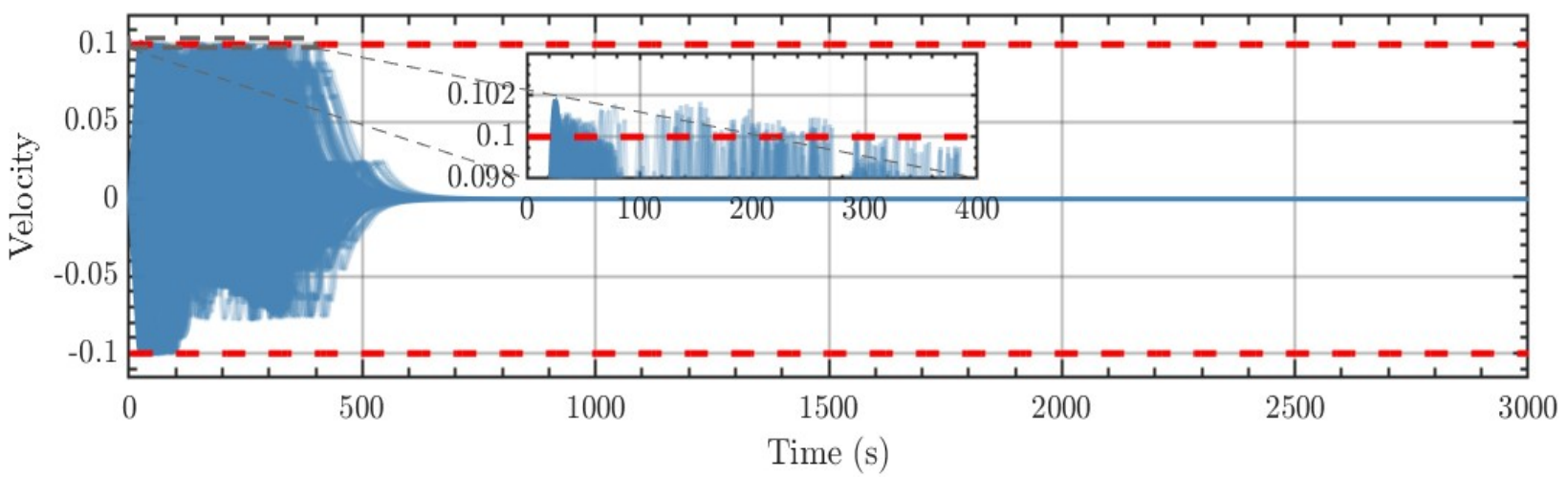}
        \label{velocity_w_cir}}
    \caption{MonteCarlo Simulation of InterSat-30}
    \label{MonteCarloSimulation}
\end{figure}

\begin{figure}[hbt!]
    \centering
    \subfigure[CBF value]{
        \centering
        \includegraphics[scale=0.3]{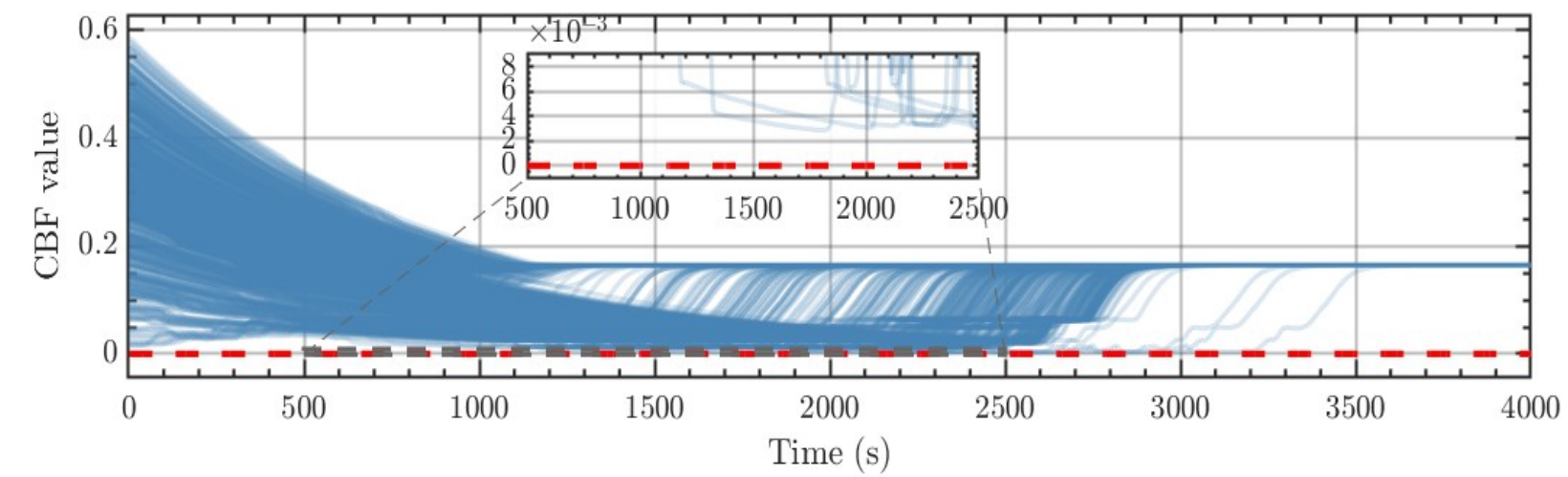}
        \label{cbf_w_cir_iss}}
    \centering
    \subfigure[CLF value]{
        \centering
        \includegraphics[scale=0.3]{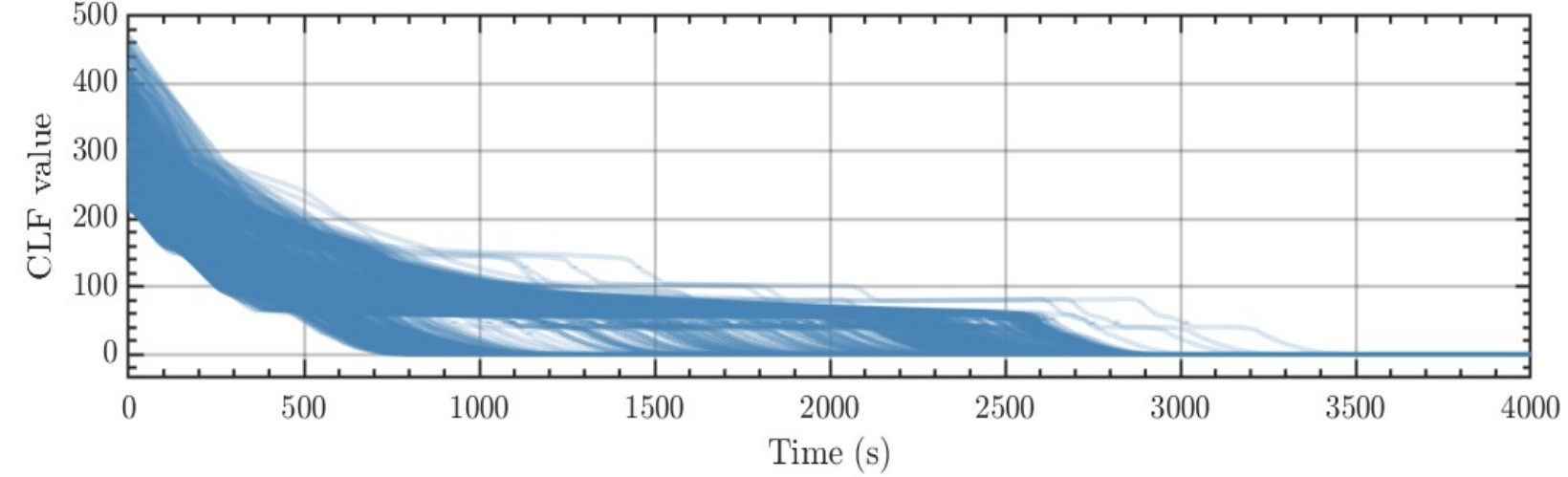}
        \label{clf_w_cir_iss}}
    \centering
    \subfigure[Position error]{
    \centering
    \includegraphics[scale=0.3]{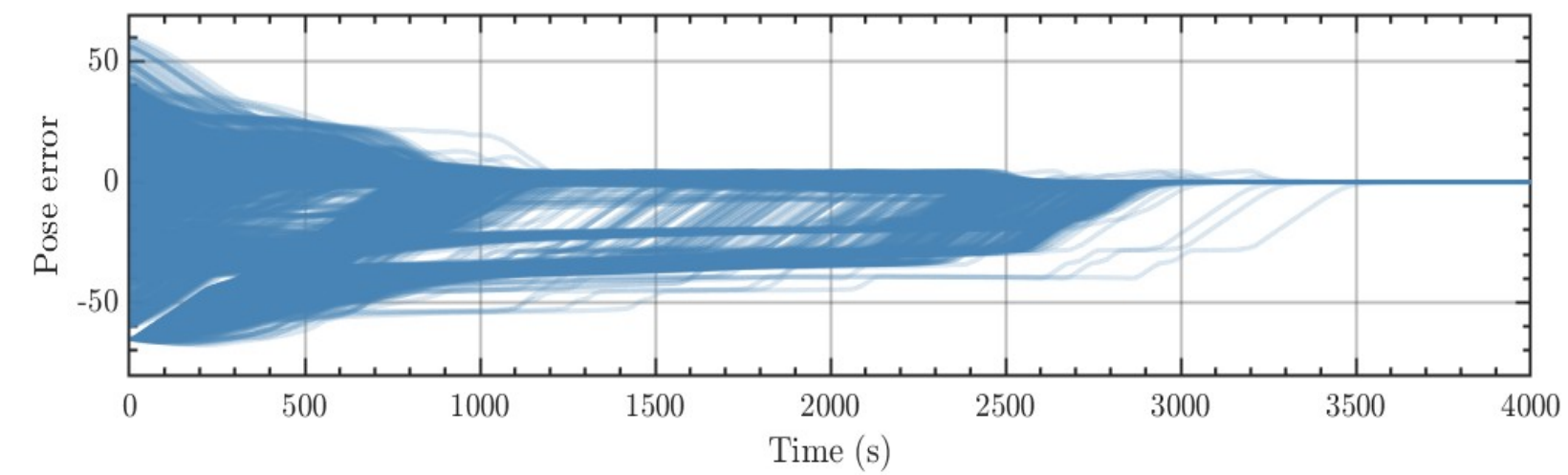}
    \label{pose_error_w_cir_iss}}
    \subfigure[Velocity]{
        \centering
        \includegraphics[scale=0.3]{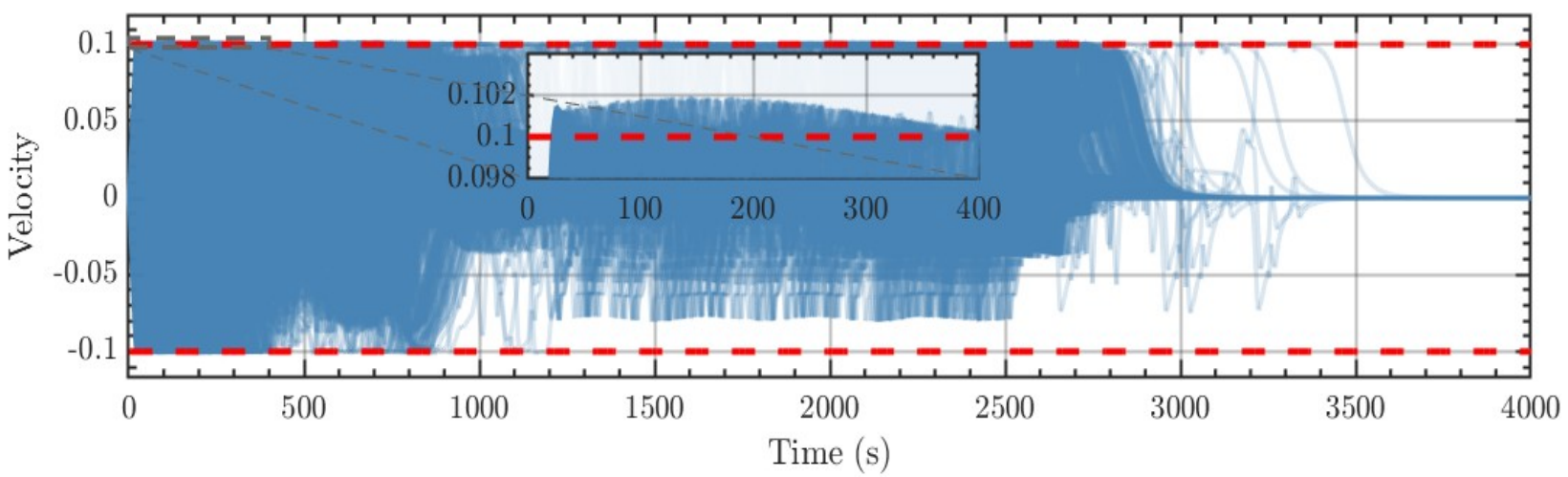}
    \label{velocity_w_cir_iss}}
    \caption{MonteCarlo Simulation of ISS}
    \label{MonteCarloSimulation_iss}
\end{figure}



\section{CONCLUSIONS}
\label{conclusion}
This paper proposes a safe robust control framework for spacecraft proximity operations in the presence of external disturbances and complex target geometries. Leveraging INRs, a neural SDF is learned from point cloud data through improved implicit geometric regularization. This learned SDF tends to over-approximate the target geometries and provides the necessary distance and gradient information for safe controller design. Based on this characterization, a two-layer safe robust control framework is developed. In the top layer, a CCBF-SOCP is designed to generate a reference velocity that guarantees safety. The bottom layer contains a DO and an improved CBF that considers the DO estimation error, which enables the construction of a smooth safety filter and enhances the robustness to external disturbances. Simulations verify the effectiveness of the proposed framework. Future research will extend this work to non-cooperative targets. For spacecraft that cannot be modeled a priori, the key challenge is to perform combat missions based on real-time local point cloud information. Addressing this challenge will be the focus of future research.

\section{ACKNOWLEDGMENT}
\label{acknowledgment}
\bibliographystyle{IEEEtran}
\bibliography{ref}

\begin{IEEEbiography}
    [{\includegraphics[width=1in,height=1.25in,clip,keepaspectratio]{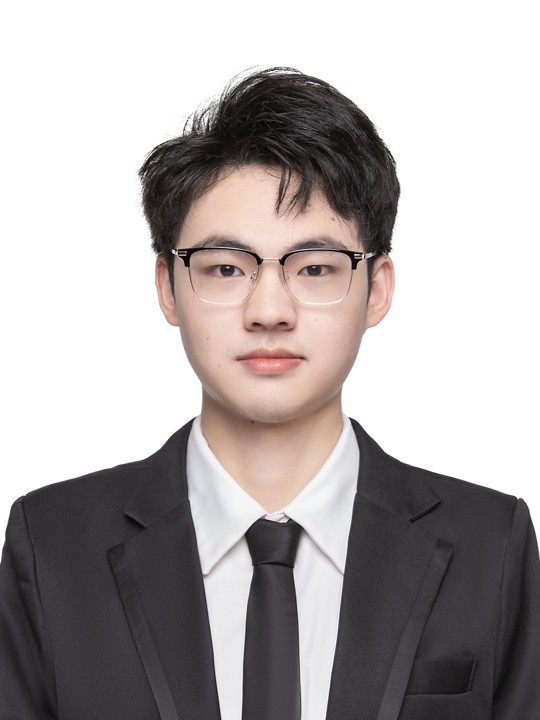}}]
	{Hang Zhou}{\space} received the B.S. dgree from the School of Information Engineering of Zhejiang University of Technology, Hangzhou, China, in 2023. He is currently pursing the Ph.D. degree in Aeronautical and Astronautical Science and Technology at Zhejiang University, Hangzhou, China. His main research interests include spacecraft 6-DOF control, spacecraft safety critical control, spacecraft trajectory optimization.
\end{IEEEbiography}

\begin{IEEEbiography}
	[{\includegraphics[width=1in,height=1.25in,clip,keepaspectratio]{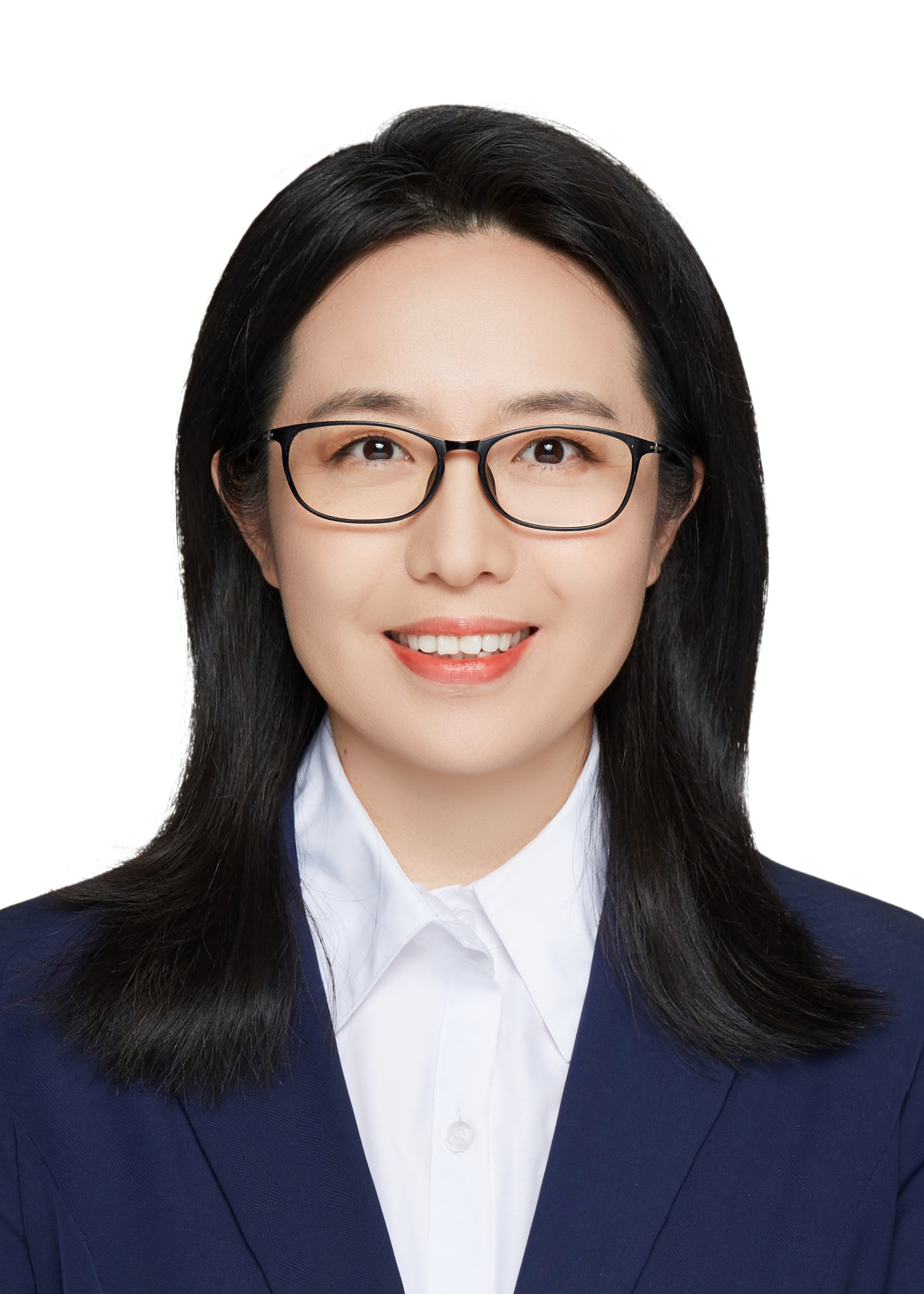}}]
	{Tao Meng}{\space}
	received the B.S. degree in Electronic science and technology, Zhejiang University, Hangzhou, China, in 2004, the M.S. degree in Electronic science and technology, Zhejiang University, Hangzhou, China, in 2006, and the Ph.D. degree in Electronic science and technology, Zhejiang University, Hangzhou, China, in 2009. She is currently a Professor at the School of Aeronautics and Astronautics. Her research interest includes attitude control, orbital control, and constellation formation control of micro-satellite.
\end{IEEEbiography}

\begin{IEEEbiography}
	[{\includegraphics[width=1in,height=1.25in,clip,keepaspectratio]{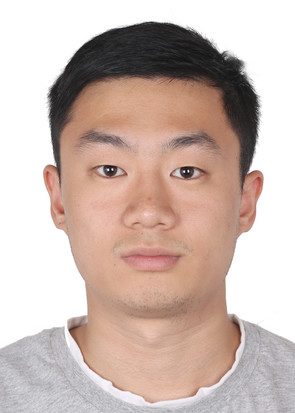}}]
	{Kun Wang}{\space}
	received the B.S. degree from College of Control Science and Engineering, Zhejiang University, Hangzhou, China, in 2019 and the Ph.D. degree in Aeronautical and Astronautical Science and Technology in Zhejiang University, Hangzhou, China, in 2024. He is currently a Post-Doctoral of the Hangzhou Innovation Institute, Beihang University, Hangzhou, China. 
	His research interests include spacecraft 6-DOF control, spacecraft safety critical control and spacecraft formation control. 
\end{IEEEbiography}

\begin{IEEEbiography}
	[{\includegraphics[width=1in,height=1.25in,clip,keepaspectratio]{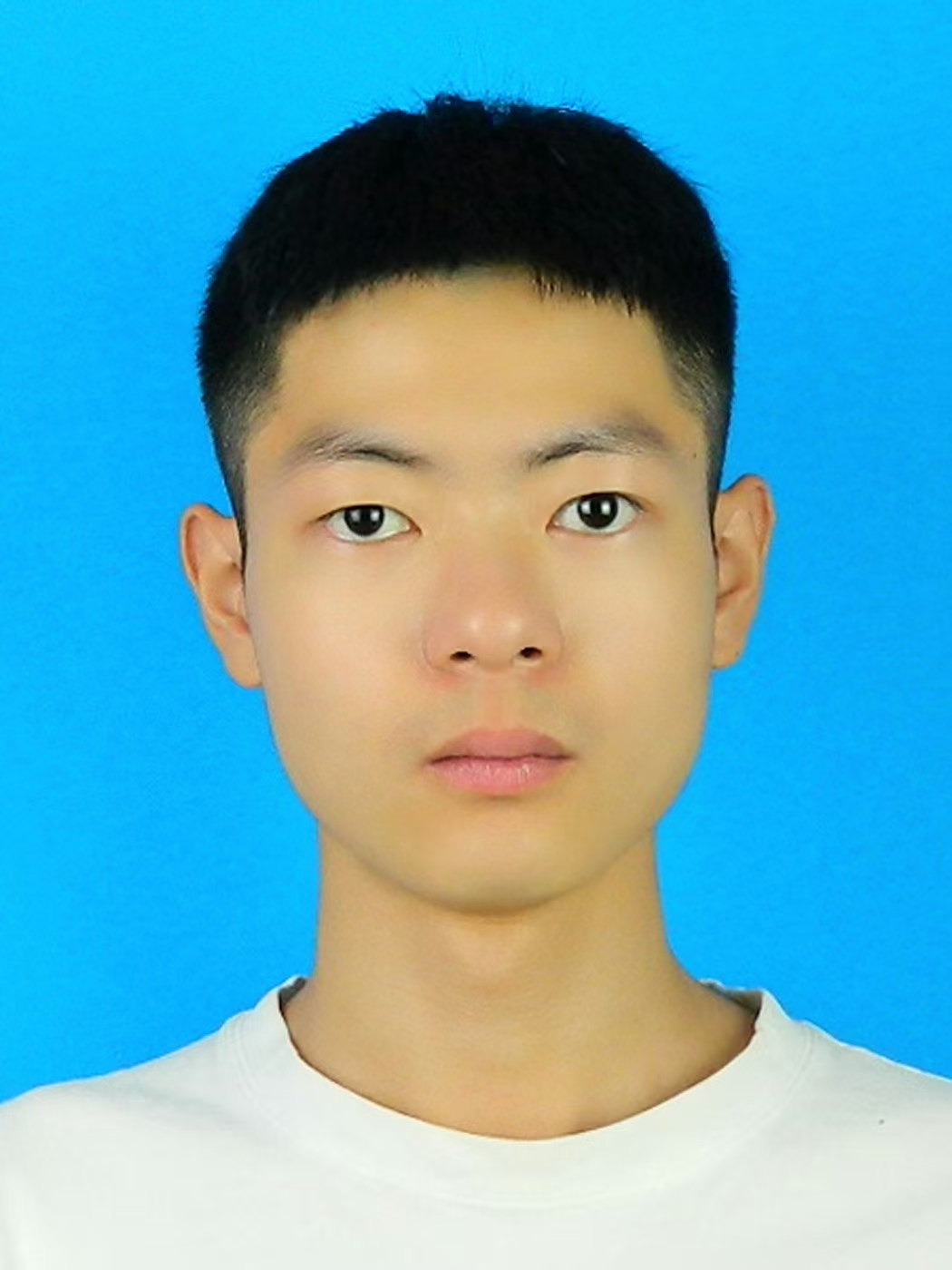}}]
	{Chengrui Shi}{\space}	 
    received the B.S. degree from College of Control Science and Engineering, Zhejiang University, Hangzhou, China, in 2024. He is currently pursuing the Maser's degree in Aeronautical and Astronautical Science and Technology at Zhejiang University, Hangzhou, China. His research interests include data-driven and safety-critical control and its application in space autonomous systems.
\end{IEEEbiography}

\begin{IEEEbiography}
    [{\includegraphics[width=1in,height=1.25in,clip,keepaspectratio]{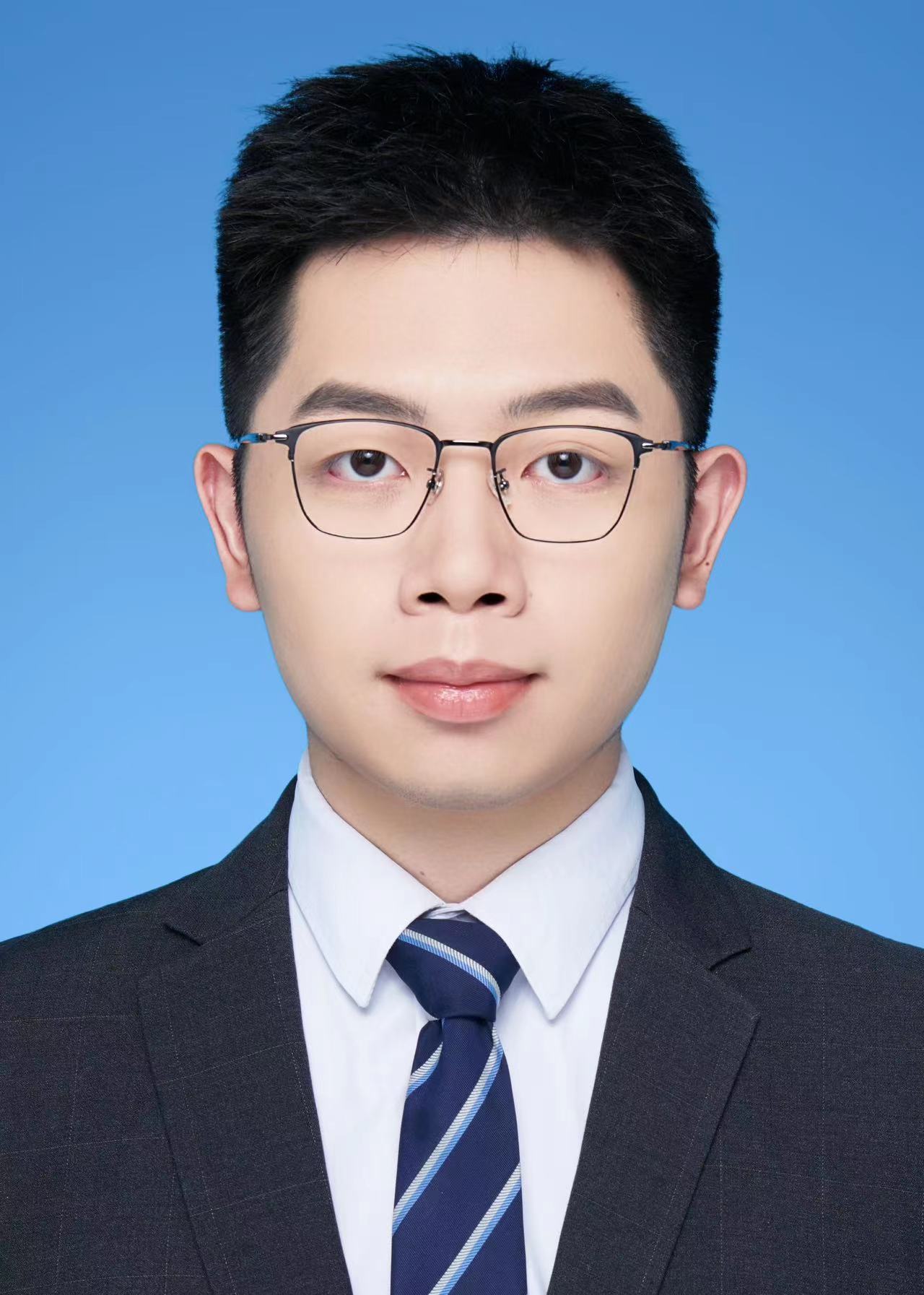}}]{Renhao Mao}{\space}
    received the B.S. dgree from the School of Astronautics, Northwestern Polytechnical University, Xi'an, China, in 2023. He is currently pursing the Ph.D. degree in Electronic Science and Technology at Zhejiang University, Hangzhou, China. His main research interests include space robots, data driven control, and machine learning.
\end{IEEEbiography}

\begin{IEEEbiography}
	[{\includegraphics[width=1in,height=1.25in,clip,keepaspectratio]{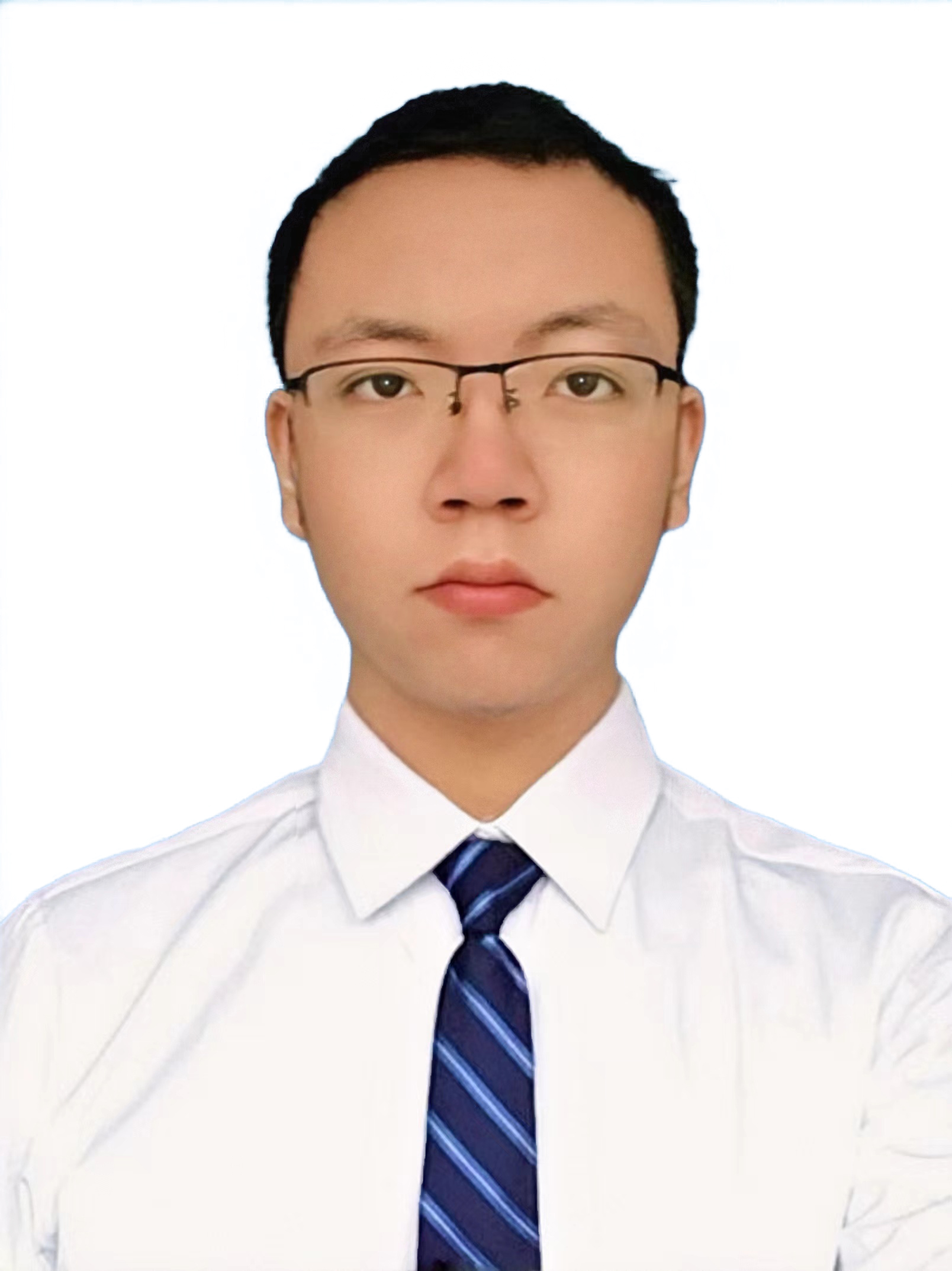}}]
	{Weijia Wang}{\space} 
	received the B.S. degree in Aerospace Engineering from the University of Electronic Science and Technology of China, Chengdu, China, in 2020. He is currently pursuing the Ph.D. degree in Aeronautical and Astronautical Science and Technology at Zhejiang University, Hangzhou, China. His research interests include model predictive control, safe critical control, and learning-based control for 6-DOF spacecraft formation.
\end{IEEEbiography}

\begin{IEEEbiography}
	[{\includegraphics[width=1in,height=1.25in,clip,keepaspectratio]{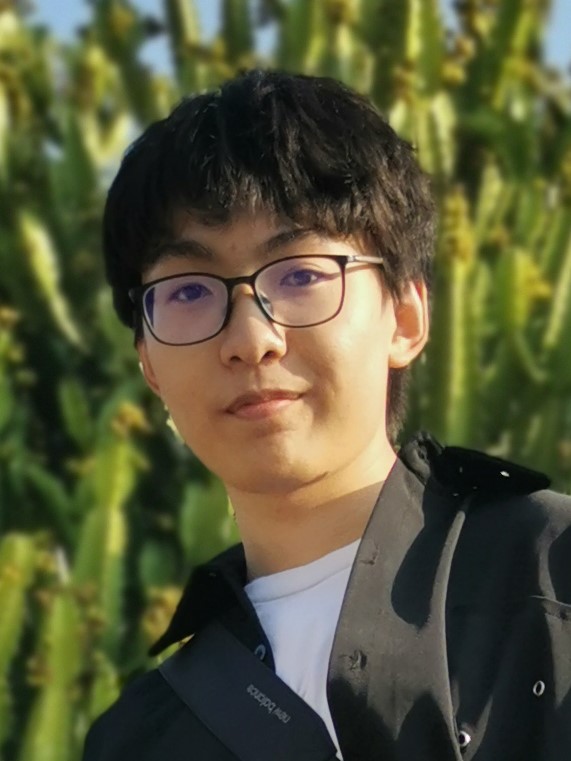}}]
	{Jiakun Lei}{\space} 
	received the B.S. degree in Automatic Control, from the University of Electronic Science and Technology of China (UESTC), Chengdu, China, in 2019; Ph.D. degree in Aeronautic and Astronautic Science and Technology, from the Zhejiang University, Hangzhou, China, in 2024. He is now a post doc of Hangzhou International Innovation Institute, Beihang University, Hangzhou, China. His research interests include constrained spacecraft control, Autonomous System, adaptive control, Safety-critical control and their application in Aerospace Engineering
\end{IEEEbiography}

\end{document}